\documentclass[smallcondensed]{svjour3} 
  \smartqed
  \journalname{Journal of Automated Reasoning}

\usepackage{amssymb,amsmath}

\usepackage{tikz}
  \usetikzlibrary{calc}
  \usetikzlibrary{shadows}
\usepackage{caption}
\usepackage{subcaption}

\usepackage[final]{listings}
  \lstdefinestyle{vcc}{
    language=C,
    basicstyle=\sffamily,
    columns=fullflexible,
    keepspaces,
    mathescape=true,
    morekeywords={ghost, invariant, requires, ensures, maintains,
      returns, reads, writes, assert, assume, record, def, datatype,
      decreases, dynamic_owns, forall, exists, lambda, 
      bool, integer, natural, true, false},
    literate=
      {forall}{$\forall$}1
      {exists}{$\exists$}1
      {==>}{$\longrightarrow$}2
      {<-->}{$\longleftrightarrow$}2
      {!!}{$\lnot$}1
      {&&&}{$\land$}1
      {|||}{$\lor$}1
      {!!=}{$\neq$}1
      {>==}{$\ge$}1
      {<==}{$\le$}1}
  \lstnewenvironment{vcclst}{\lstset{style=vcc,
    basicstyle=\sffamily\footnotesize}}{}
  \lstnewenvironment{vccquote}{\lstset{style=vcc,
    basicstyle=\sffamily\footnotesize, xleftmargin=2em}}{}
  \newcommand{\inlinevcc}[1]{\lstinline[style=vcc]^#1^}

  \lstdefinestyle{isa}{
    columns=fullflexible,
    keepspaces,
    mathescape=true,
    morekeywords={locale,fixes,assumes,shows,and,lemma,definition,in,
      type_synonym,where},
    literate=
      {"}{}0
      {\\<forall>}{$\forall$}1
      {\\<exists>}{$\exists$}1
      {\\<Longrightarrow>}{$\Longrightarrow$}2
      {\\<longrightarrow>}{$\longrightarrow$}2
      {\\<Rightarrow>}{$\Rightarrow$}1
      {\\<longleftrightarrow>}{$\longleftrightarrow$}2
      {\\<And>}{$\bigwedge$}1
      {\\<and>}{$\land$}1
      {\\<or>}{$\lor$}1
      {\\<in>}{$\in$}1
      {\\<noteq>}{$\neq$}1
      {\\<le>}{$\le$}1
      {\\<ge>}{$\ge$}1
      {\\<subseteq>}{$\subseteq$}1
      {\\<lbrakk>}{$\mathopen{\lbrack\mkern-3mu\lbrack}$}1
      {\\<rbrakk>}{$\mathclose{\rbrack\mkern-3mu\rbrack}$}1
      {\\<infinity>}{$\infty$}1
      {\\<mu>}{$\mu$}1
      {\\<Sum>}{$\sum$}1
      {'a}{$\alpha$}1
      {'b}{$\beta$}1
      {parent_num_assms}{$\mathit{parent\dash{}num\dash{}assms}$}3
      {edges}{$\mathit{edges}$}5
      {verts}{$\mathit{verts}$}5
      {start}{$\mathit{start}$}5
      {target}{$\mathit{target}$}6
      {num}{$\mathit{num}$}3
      {just}{$\mathit{just}$}3
      {dist}{$\mathit{dist}$}4
      {no_path}{$\mathit{no\dash{}path}$}3
      {pos_cost}{$\mathit{pos\dash{}cost}$}3
      {s_in_G}{$\mathit{s\dash{}in\dash{}G}$}3
      {sstart_val}{$\mathit{start\dash{}val}$}3
      {parent_edge}{$\mathit{parent\dash{}edge}$}3
      {digraph}{$\mathit{digraph}$}1
      {pseudo_digraph}{$\mathit{pseudo\dash{}digraph}$}1
      {pre_graph}{$\mathit{pre\dash{}graph}$}1
      {connected_components_locale}{$\mathit{connected\dash{}components\dash{}locale}$}3
      {shortest_path_pos_cost}{$\mathit{shortest\dash{}path\dash{}pos\dash{}cost}$}1
      {basic_just_sp}{$\mathit{basic\dash{}just\dash{}sp}$}1
      {basic_sp}{$\mathit{basic\dash{}sp}$}1
      {dist_ge_\\<mu>}{$\mathit{dist\dash{}ge\dash{}\mu}$}1
      {disjoint_edges}{$\mathit{disjoint\dash{}edges}$}1
      {matching}{$\mathit{matching}$}1
      {\\matching}{matching}1
      {maxM}{$\mathit{maxM}$}1
      {matching_locale}{$\mathit{matching\dash{}locale}$}1
      {OSC}{$\mathit{OSC}$}1
      {\\OSC}{OSC}1
      {weight}{$\mathit{weight}$}1
      {\\weight}{weight}1
      {card}{$\mathit{card}$}1
      {label}{$\mathit{label}$}3
      {bool}{$\mathit{bool}$}3
      {set}{$\mathit{set}$}3
      {enat}{$\mathit{enat}$}3
      {nat}{$\mathit{nat}$}3
      {ereal}{$\mathit{ereal}$}3
      {real}{$\mathit{real}$}3
      {option}{$\mathit{option}$}6
      {None}{$\mathit{None}$}4
      {Some}{$\mathit{Some}$}4
      {\\0}{$0$}1
      {\\1}{$1$}1
      {\\2}{$2$}1
      {\\G}{$G$}1
      {\\M}{$M$}1
      {\\L}{$L$}1
      {LV}{$LV$}1
      {\\N}{$N$}1
      {\\V}{$V$}1
      {\\e}{$e$}1
      {\\i}{$i$}1
      {\\f}{$f$}1
      {e1}{$e_1$}1
      {e2}{$e_2$}1
      {\\v}{$v$}1
      {\\u}{$u$}1
      {\\r}{$r$}1
      {\\s}{$s$}1
      {\\c}{$c$}1}
  \lstnewenvironment{isalst}{\lstset{style=isa,
    basicstyle=\footnotesize}}{}
  \newcommand{\inlineisa}[1]{\lstinline[style=isa]^#1^}

  \newcounter{listing}
  \newcounter{savefigure}
  \newenvironment{listing}
    {\setcounter{savefigure}{\value{figure}}
     \setcounter{figure}{\value{listing}}
     \renewcommand{\figurename}{Listing}\begin{figure}}
    {\end{figure}
     \setcounter{figure}{\value{savefigure}}
     \addtocounter{listing}{1}}

\newcommand{\ignore}[1]{}  

\newcommand{\eg}{e.g.,\ }
\newcommand{\ie}{i.e.,\ }
\newcommand{\cf}{cf.\ }
\newcommand{\wrt}{w.r.t.\ }

\newcommand{\Pre}{{\varphi}}
\newcommand{\Post}{{\psi}}
\newcommand{\Wit}{{\cal{W}}}
\newcommand{\Nat}{{\mathbb{N}}}
\newcommand{\abs}[1]{\left\lvert #1 \right\rvert}

\newcommand{\dash}{\mbox{-}}

\newcommand{\iX}{i_X}
\newcommand{\iY}{i_Y}
\newcommand{\iW}{i_W}

\usepackage{pdfsetup}
\begin{document}

\title{A Framework for the Verification of Certifying Computations}

\author{
  Eyad~Alkassar \and
  Sascha~B\"ohme \and
  Kurt~Mehlhorn \and
  Christine~Rizkallah}

\authorrunning{E. Alkassar, S. B\"ohme, K. Mehlhorn, C. Rizkallah}

\institute{
  E. Alkassar \at Fachbereich Informatik, Universit\"at des Saarlandes, Germany \and
  S. B\"ohme \at Institut f\"ur Informatik, Technische Universit\"at
    M\"unchen, Germany \and
  K. Mehlhorn \and C. Rizkallah \at Max-Planck-Institut f\"ur Informatik,
    Saarbr\"ucken, Germany}


\maketitle

\begin{abstract}
Formal verification of complex algorithms is challenging. Verifying their
implementations goes beyond the state of the art of current
automatic verification
tools and usually involves intricate mathematical
theorems.  Certifying algorithms compute in addition to each output a
witness certifying that the output is correct.  A checker for such a
witness is usually much simpler than the original algorithm -- yet it is
all the user has to trust.  The verification of checkers is feasible with
current tools and leads to computations that can be completely trusted.  We
describe a framework to seamlessly verify certifying computations. We use
the automatic verifier VCC for establishing the correctness of the checker and
the interactive theorem prover Isabelle/HOL for high-level mathematical
properties of algorithms.  We demonstrate the effectiveness of our approach
by presenting the verification of typical examples of the industrial-level
and widespread algorithmic library LEDA.
\end{abstract}

\section{Introduction}

One of the most prominent and costly problems in software engineering is
correctness of software. In this article, we are concerned with software
for difficult algorithmic problems, \eg in the domain of graphs. The
algorithms for such problems are complex; formal verification of the
resulting programs is not tractable, which explains
why few graph algorithms have been verified. We show how to obtain
\emph{formal instance correctness}, \ie formal proofs that outputs for
particular inputs are correct. We do so by combining the concept of
certifying algorithms with methods for code verification and theorem
proving.

A \emph{certifying
algorithm}~\cite{blum-kannan:1989:check,sullivan-masson:1990:certification,mcconnell-etal:2011:certifying} 
produces with each output a \emph{certificate} or \emph{witness} that the
\emph{particular output} is correct. The accompanying \emph{checker} for a certifying algorithm with input $x$, output $y$, and witness $w$ takes as input the triple 
$(x,y,w)$ and accepts the triple if $w$ proves that $y$ is a
correct output for input $x$. Otherwise, the checker rejects the output or
witness as buggy\footnote{Throughout the paper, we say the checker \emph{accepts} if the checker returns $\mathit{True}$; otherwise, we say it \emph{rejects}.}. 

Figure~\ref{fig:black-box-view} contrasts a standard
algorithm with a certifying algorithm for IO-behavior $(\Pre,\Post)$. An
algorithm for IO-behavior $(\Pre,\Post)$ receives an input $x$ satisfying a
precondition $\Pre(x)$ and is supposed to deliver an output $y$ satisfying
the postcondition $\Post(x,y)$. We call such a $y$ a \emph{correct output}.
If the input does not satisfy the precondition, the result of the
computation is unspecified. A user of a standard algorithm has, in general,
no means of knowing that $y$ is a correct output and has not been compromised
by a bug. In contrast, if the accompanying checker of a certifying algorithm accepts, the user may
proceed with the complete confidence that output $y$ has not been compromised
by a bug. If the checker rejects, either $y$ is incorrect or $w$ is not a
proof of the correctness of $y$.

\begin{figure}[t]
\centering
\begin{tikzpicture}
\tikzstyle{box}=[draw, fill=white, drop shadow={opacity=0.8},
  inner xsep=8pt, minimum height=1.5cm];
\node[box,text width=2.5cm,text centered] (P) at (0,5) {Program for IO-behavior $(\Pre,\Post)$};
\draw[-latex] (-2.5,5) -- node[below] {$x$} (P);
\draw[-latex] (P) -- node[below] {$y$} (2.5,5);

\node[box, text width=2.5cm, text centered] (cP) at (-2,2)
  {Certifying program for IO-behavior $(\Pre,\Post)$};
\node[box] (C) at (2,2) {Checker $C$};
\draw[-latex]
  (-5,3.2) -- (0.2,3.2) |- node[below, pos=0.7] {$x$}
  ($(C.west) + (0,0.5)$);
\draw[-latex] (-4.3,3.2) |- node[below, pos=0.7] {$x$} (cP.west);
\draw[-latex] (cP.east) -- node[below] {$y$} (C.west);
\draw[-latex]
  ($(cP.east) - (0,0.5)$) -- node[below] {$w$}
  ($(C.west) - (0,0.5)$);
\draw[-latex] ($(C.east) + (0,0.5)$) -- node[above] {accept $y$} +(2,0);
\draw[-latex] ($(C.east) - (0,0.5)$) -- node[below] {reject} +(2,0);
\end{tikzpicture}
\caption{The top figure shows the I/O behavior of a conventional program
  for IO-behavior $(\Pre,\Post)$. The user feeds an input $x$ satisfying
  $\Pre(x)$ to the program, and the program returns an output $y$ satisfying
  $\Post(x,y)$. A certifying algorithm for IO-behavior $(\Pre,\Post)$
  computes $y$ and a witness $w$. The checker $C$ accepts the triple
  $(x,y,w)$ if and only if $w$ is a valid witness for the postcondition
  $\Post(x,y)$, i.e., it proves $\Post(x,y)$.}
  \label{fig:black-box-view} 
\end{figure}
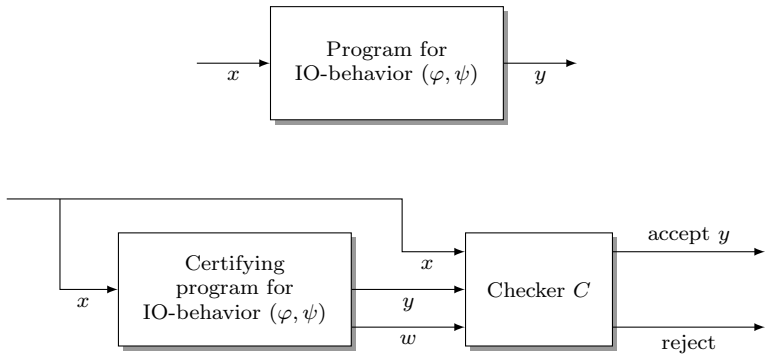


\ignore{We illustrate the concept certifying algorithm by a simple
  example. A graph is bipartite if the vertices can be colored by colors red
  and blue such that the endpoints of every edge have distinct colors. The
  output of a non-certifying algorithm is either YES or NO. A certifying
  algorithm may output a two-coloring in the YES-case and an odd-length cycle
  contained in the graph in the NO-case. An odd-length cycle can clearly not be
  two-colored and hence any graph containing an odd-length cycle cannot be
  two-colored. The checker proceeds as follows. In the YES-case, it iterates
  over all edges of the graph and checks that the endpoints have distinct
  colors. In the NO-case, it checks that all edges of the cycle are
  present in the graph and that the cycle has odd length. }

We illustrate the concept of certifying algorithms with an example. The
greatest common divisor of two nonnegative integers $a$ and $b$, not both zero, is
the largest integer $g$ that divides $a$ and $b$. We write $g =
\mathit{gcd}(a, b)$. The extended Euclidean algorithm
is a certifying algorithm for greatest common divisor. In addition to the output
$g = gcd(a, b)$, it also computes as a witness\footnote{It is known that such
  integers $s$ and $t$ always exist.} integers $s$ and
$t$ such that $g = s\cdot a+t \cdot b$. The checker checks that $g$ divides $a$
and $b$ and that $g = s\cdot a+t \cdot b$. Why does this prove that $g$ is the
greatest common divisor of $a$ and $b$? Consider any integer $d$ that divides $a$ and
$b$. Then $g = s\cdot a+t \cdot b = (s \cdot (a/d) + t \cdot (b/d)) \cdot d$, 
and hence, $d$ divides $g$.

Certifying algorithms are a key design principle of the algorithmic library
LEDA~\cite{melhorn-naeher:1999:leda}: Checkers are an integral part of the
library and may (optionally) be invoked after every execution of a LEDA
algorithm.  The adoption of the principle greatly improved the reliability of
the library.  However, how can one be sure that the checker programs are
correct? The third author used to answer:
``Checkers are simple programs with little algorithmic complexity. Hence, one
may assume that their implementations are correct." 
We give a better answer in this paper.

We take the certifying-algorithms approach a step further by developing a
methodology to verify the checkers. 
We demonstrate it on three examples:
the connectivity of graphs, single-source shortest paths in graphs with
nonnegative edge weights, and maximum cardinality matching in graphs.
The latter is one of the more complex algorithms in LEDA. The
description of the algorithm and its implementation
in~\cite{melhorn-naeher:1999:leda} comprises 15 pages. In contrast, the
checker fits on one page.  Our formalization
effort has revealed that the checker
program in LEDA is incorrect in that it does not check that the matching of a graph is a subset of the edges of the graph (see Section~\ref{sec:case-study-matching} for a definition of matching).

We introduce our methodology in Section~\ref{sec:outline-of-approach} and
give detailed case studies in Section~\ref{sec:case-studies} before
evaluating our approach and the obtained results in
Section~\ref{sec:evaluation}.  In Section~\ref{sec:tools}, we survey the
verification tools VCC and Isabelle/HOL. Section~\ref{sec:related-work}
discusses related work, and Section~\ref{sec:conclusions} offers
conclusions. The companion web page
\url{http://www21.in.tum.de/~boehmes/certifying\_computations.html}
contains additional material, in particular, the program listings including
VCC annotations and the Isabelle/HOL proofs.

This article is a revised and extended version of a paper published by the
same authors at CAV~2011~\cite{alkassar:2011:vcc}.  We have added two case
studies to underline the feasibility and elegance of our approach.
Moreover, we have strengthened and simplified our approach. We now prove
the total
correctness of the checkers instead of only partial correctness. Furthermore, we 
establish that the checker accepts a triple $(x,y,w)$ if and only if $w$ is
a valid witness for output $y$. Previously, we only proved the if direction. 
The simplification results from having only one kind of specification in VCC 
instead of two. Furthermore, a new version of VCC allows for shorter specifications and proofs.

\section{Outline of Methodology}
  \label{sec:outline-of-approach}

We consider algorithms that take an input from a set $X$ and produce an
output in a set $Y$ and a witness in a set $W$. The input $x \in X$ is
supposed to satisfy a precondition $\Pre(x)$, and the input together with
the output $y \in Y$ is supposed to satisfy a postcondition $\Post(x,y)$. 
A \emph{witness predicate} for a specification with precondition~$\Pre$ and
postcondition $\Post$ is a predicate $\Wit \subseteq X \times Y \times W$,
where $W$ is a set of witnesses
with the following \emph{witness property}:

\begin{equation}
\Pre(x) \land \Wit(x,y,w) \longrightarrow \Post(x,y).
\label{witness-property}
\end{equation}

\noindent
In contrast to algorithms that work on abstract sets $X$, $Y$, and $W$,
the implementing programs operate on concrete representations of
abstract objects.  We use $\overline{X}$, $\overline{Y}$, and
$\overline{W}$ for the set of representations of objects in $X$, $Y$, and
$W$, respectively, and assume the mappings $\iX: \overline{X} \rightarrow X$,
$\iY: \overline{Y} \rightarrow Y$, and $\iW: \overline{W} \rightarrow
W$.

We illustrate these definitions through the example from the introduction. 
In the case of greatest common divisors, $X$ and $W$ are the set of pairs of
integers, and $Y$ is the set of integers. For input $(a,b)$, output $g$ and
witness $(s,t)$, the precondition $\Pre( (a,b) )$ states that the inputs
$a$ and $b$ are nonnegative integers and that at least one of them is not
zero. The postcondition $\Post((a,b), g)$ states that $g =
\mathit{gcd}(a,b)$. The witness predicate $\Wit((a,b),g,(s,t))$ states that
$g = s a+t b$ and that $g$ divides $a$ and $b$. A typical
representation of integers in implementations is via bitstrings. Hence, $\overline{X}$ and
$\overline{W}$ are each the set of pairs of bit strings, and
$\overline{Y}$ is the set of bit
strings. The mappings $\iX$, $\iY$, and $\iW$ map {(pairs of)} bit strings to
the corresponding {(pairs of)} integers.

The checker program $C$ receives a triple
$(\overline{x},\overline{y},\overline{w})$ and is supposed to
check whether
it fulfills the witness property. More precisely, let  $x =
\iX(\overline{x})$, $y = \iY(\overline{y})$, and $w = \iW(\overline{w})$.
If $\neg \Pre(x)$, $C$ may do anything (run forever or halt with an
arbitrary output). If $\Pre(x)$, $C$ must halt and either accept or reject.
It is required to accept if $\Wit(x,y,w)$
holds and is required to reject otherwise.  

\begin{description}
\item[Checker Correctness:]
  We need to prove that $C$ checks the witness predicate, 
  assuming that the precondition\footnote{We stress that the checker has the same
  precondition as the algorithm.} holds, \ie on input $(\overline{x},
  \overline{y}, \overline{w})$ and with $x = \iX(\overline{x})$, $y =
  \iY(\overline{y})$, and $w = \iW(\overline{w})$:

  \begin{enumerate} 
  \item If $\Pre(x)$, $C$ halts.
  \item If $\Pre(x)$ and $\Wit(x,y,w)$, $C$ accepts $(x,y,w)$, and if $\Pre(x)$ and
    $\neg \Wit(x,y,w)$, $C$ rejects the triple. 
  \end{enumerate}

\item[Witness Property:]
  We need to prove implication~(\ref{witness-property}).
\end{description}

In our running example, the witness property is
\[  a \ge 0 \wedge b \ge 0 \wedge a + b > 0 \wedge g | a \wedge g | b \wedge g
= a \cdot s + b \cdot t \implies g = \mathit{gcd}(a,b).  \]
Here $a$, $b$, $g$, $s$ and $t$ are assumed to be nonnegative integers.

\begin{theorem}
Assume that the proof obligations are discharged.  Let
$(\overline{x},\overline{y},\overline{w}) \in \overline{X} \times
\overline{Y} \times \overline{W}$ and let  $x = \iX(\overline{x})$, $y =
\iY(\overline{y})$, and $w = \iW(\overline{w})$. 

If $C$ accepts a triple $(\overline{x},\overline{y},\overline{w})$,
$\Pre(x) \longrightarrow \Post(x,y)$ by a formal proof.  If $C$ rejects a
triple $(\overline{x},\overline{y},\overline{w})$, $\neg \Pre(x) \lor \neg
\Wit(x,y,w)$ by a formal proof.
\end{theorem}

\begin{proof}
If $C$ accepts $(\overline{x},\overline{y},\overline{w})$, we have $\Pre(x)
\longrightarrow \Wit(x,y,w)$ by the correctness proof of $C$.  Then by
(\ref{witness-property}) we have a formal proof for $\Pre(x)
\longrightarrow \Post(x,y)$. Conversely, if $C$ rejects the triple,  the
correctness proof of $C$ establishes  $\neg \Pre(x) \lor \neg\Wit(x,y,w)$.
\qed
\end{proof}


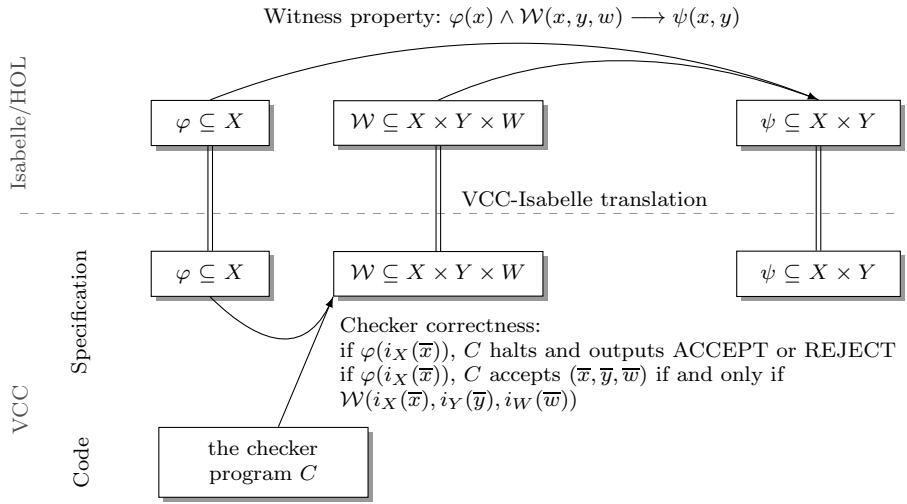
\begin{figure}[t]
\centering
\begin{tikzpicture}
\tikzstyle{box}=[text centered, inner sep=5pt, draw, fill=white,
  drop shadow={opacity=0.8}]
\tikzstyle{box1}=[box, text width=1.2cm]
\tikzstyle{box2}=[box, text width=1.8cm]
\tikzstyle{box3}=[box, text width=2.4cm]

\node[rotate=90, black!70] at (-2,4) {Isabelle/HOL};
\node[rotate=90, black!70] at (-2,0.2) {VCC};
\draw[dashed, black!40] (-2,2.8) -- (9.5,2.8);

\node[rotate=90, text width=2cm, text centered] at (-1.2,-0.5) {Code};
\node[rotate=90, text width=2cm, text centered] at (-1.2,1.5) {Specification};

\node[box3] (code) at (1.2,-0.5)
  {the checker program $C$};
\node[box1] (cpre) at (0.5,2) {$\Pre \subseteq X$};
\node[box3] (cwit) at (3.5,2) {$\Wit \subseteq X \times Y \times W$};
\node[box1] (ipre) at (0.5,4) {$\Pre \subseteq X$};
\node[box3] (iwit) at (3.5,4) {$\Wit \subseteq X \times Y \times W$};
\node[box2] (icon) at (8.5,4) {$\Post \subseteq X \times Y$};
\node[box2] (ccon) at (8.5,2) {$\Post \subseteq X \times Y$};

\draw ($(cpre.north) - (0.03,0)$) -- ($(ipre.south) - (0.03,0)$);
\draw ($(cpre.north) + (0.03,0)$) -- ($(ipre.south) + (0.03,0)$);
\draw ($(cwit.north) - (0.03,0)$) -- ($(iwit.south) - (0.03,0)$);
\draw ($(cwit.north) + (0.03,0)$) -- ($(iwit.south) + (0.03,0)$);
\draw ($(icon.south) - (0.03,0)$) -- ($(ccon.north) - (0.03,0)$);
\draw ($(icon.south) + (0.03,0)$) -- ($(ccon.north) + (0.03,0)$);
\draw (cpre.south) .. controls (1.2,1.0) and (1.8,0.9) .. (cwit.south west);
\draw[-latex] (code) -- (cwit.south west);
\node[text width=7.5cm, anchor=west] at (2.1,0.8) {
  Checker correctness: \\
  if $\Pre(\iX(\overline{x}))$, $C$ halts and outputs 
  ACCEPT or REJECT\\
  if $\Pre(\iX(\overline{x}))$, $C$ accepts
  $(\overline{x},\overline{y},\overline{w})$ if and only if
  $\Wit(\iX(\overline{x}),\iY(\overline{y}),\iW(\overline{w}))$};
\node[anchor=west] at (3.7,3) {VCC-Isabelle translation};
\draw[-latex] (ipre.north) .. controls (3,5.3) and (6,5.3) .. (icon.north);
\draw (iwit.north) .. controls (5,5) and (6.7,5) .. (icon.north);
\node[text width=7cm] at (4.7,5.4) {
  Witness property:  $\Pre(x) \land \Wit(x,y,w) \longrightarrow \Post(x,y)$};
\end{tikzpicture}
\caption{Verification Framework and Proof Obligations.}
\label{fig:framework}
\end{figure}

We next discuss how we fulfill the proof obligations in a
\emph{comprehensive} and \emph{efficient} framework, see
Figure~\ref{fig:framework}. Comprehensive means that the final proof
formally combines (as much as possible at the syntactic level) the
correctness arguments for all levels (implementation, abstraction, and
mathematical theory).  Efficient means using the right tool for the right
task. For example, applying a general theorem prover to verify imperative
code would involve a lot of language-specific overhead and lead to less
automation; similarly, a specialized code verifier is often not powerful
enough to cover nontrivial mathematical properties.  The goals of
comprehensiveness and efficiency seem to be conflicting because different tools
usually come with different languages, axiomatization sets, etc.  Our
solution is to use second-order logic as a common interface language.

LEDA is written in C++~\cite{melhorn-naeher:1999:leda}. Our aim is to
verify code which is as close as possible to the original implementation;
by this, we demonstrate the feasibility of verifying already established
libraries written in imperative languages such as C. Thus, we verify code
with VCC~\cite{cohen-etal:2009:vcc}, an automatic code verifier for full
C. Verifying the C++ implementations remains open.
Our choosingVCC is motivated by the maturity of the tool and the provision of an
assertion language that is rich enough for our requirements.  In the
Verisoft XT project~\cite{verisoftxt}, VCC was successfully used to verify
tens of thousands of lines of C code. The assertion language offers 
ghost code and ghost types such as maps and unbounded integers.  This gives
enough expressiveness to quantify over graphs, labelings, etc., and
simplifies the translation to other proof systems.  For verifying the
mathematical part, we use Isabelle/HOL, a higher-order-logic
interactive theorem prover~\cite{nipkow-etal:2002:isabelle}, because of the
large amount of already formalized mathematics, its descriptive proof format, 
and its various automatic proof methods and tools.  In
Section~\ref{sec:tools}, we review both systems.
Figure~\ref{fig:framework} shows the workflow for verifying checkers.

\begin{description}
\item[Checker Verification:]
  The starting point is the checker code written in C.  Using VCC, we annotate
  the functions and data structures such that the witness predicate
  $\Wit$ can be established as the postcondition of the checker function.  We
  define the witness predicate and the pre- and postcondition as well as
  the mappings $\iX$, $\iY$, and $\iW$ as pure mathematical objects using
  VCC ghost types and ghost functions.
		
\item[Export to Isabelle/HOL:]
  Establishing the witness property involves, in general, mathematical reasoning 
  beyond what is conveniently done in VCC. We therefore translate the precondition, witness
  predicate, postcondition, and the abstract representations of the input,
  output, and witness from VCC to Isabelle/HOL.
   Since we formulated them as
  pure mathematical objects in VCC, this translation is purely syntactical
  and does not involve any VCC specifics. The translation could easily be
  automated. 

\item[Witness Property:]
  We prove the witness property using Isabelle/HOL. It is convenient to
  formulate this theorem on yet a higher level of abstraction and provide
  linking proofs to connect the exported VCC predicates with their
  abstracted counterparts.
\end{description}

We stress that the overall correctness theorem, \ie the witness property,
can be formulated in VCC; this is important for usability. The user of a
verified checker only has to look at its VCC specification; the fact that
we outsource the proof of the witness property to Isabelle/HOL is of no
concern to the user.  We formulate the witness property as an axiom in VCC. This
is sound since we restrict the language for describing the witness property
to second-order logic, which guarantees that we can express it equivalently
in Isabelle's higher-order logic (\cf Section~\ref{sec:tools}). More
precisely, since the VCC formulation of the witness property is valid if
and only if its translation to Isabelle is valid, and since Isabelle is
consistent, and hence, only valid statements can be proven, it is sound to
add the witness property as an axiom to VCC.

The reader may wonder why we do not formally prove the existence of a
witness:
\[
  \forall x \, y.\  \Pre(x) \land \Post(x,y) \longrightarrow
  \exists w.\ \Wit(x,y,w).
\]
The existence of a witness is part of the correctness argument of the
solution algorithm (e.g., the shortest-path algorithm, the maximum-matching
algorithm). As previously mentioned, we do not verify the solution algorithms.
 Rather, the execution of the solution
algorithm establishes the existence of a witness whenever it is called for
a specific input $\overline{x}$. It returns $\overline{y}$ and
$\overline{w}$, which we then hand to the checker $C$. In this way, we
obtain formal instance correctness without having to verify the solution
algorithm. Of course, this leaves the possibility that the solution
algorithm is incorrect and does not always provide a  $\overline{y}$ and
$\overline{w}$ such that the checker accepts $(\overline{x}, \overline{y},
\overline{w})$.

For a user, the checker is what counts. The user can trust it because it
has been formally verified. Moreover, if it accepts a triple
$(\overline{x}, \overline{y}, \overline{w})$, the user can be sure that
$y$ is a correct output, provided that $x$ satisfies the precondition of the
algorithm.  This is because the witness property has been formally
verified. If the checker rejects a triple, the user knows that either $x$
does not satisfy the precondition or $(x,y,w)$ does not satisfy the witness
predicate. The method by which $\overline{y}$ and $\overline{w}$ were
produced is of no concern to the user.

The witness property is formulated with respect to a certain IO-behavior
$(\Pre,\Post)$ and not with respect to a particular algorithm that realizes the
IO-behavior. Therefore, a checker can be used in connection with any
certifying algorithm for IO-behavior $(\Pre,\Post)$ that produces the
appropriate witnesses.

\section{Tool Overview: Isabelle/HOL and VCC}
  \label{sec:tools}

\paragraph{Isabelle/HOL}~\cite{nipkow-etal:2002:isabelle} is an interactive theorem
prover for classical higher-order logic based on Church's simply typed
lambda calculus. The system is built on top of an inference
kernel, which provides only a small number of rules to construct theorems;
complex deductions (especially by automatic proof methods) ultimately rely
on these rules only.  This approach, called LCF due to its pioneering
system~\cite{gordon-etal:1979:lcf}, guarantees the correctness of theorems proven 
in the system as long as the inference kernel is correct.  Isabelle/HOL comes with 
a rich set of already formalized theories, among which are natural numbers and 
integers as well as sets, finite sets, and---as a recent
addition~\cite{noschinski:2011:graph-library}---graphs.  New types can, for example, 
be introduced by defining them as records (isomorphic to tuples with named
update and selector functions). New constants can be
introduced, for example, via definitions relative to already existing
constants.

Proofs in Isabelle/HOL can be written in a style close to that of
mathematical textbooks.  The user structures the proof, 
and the system fills in the gaps with its automatic proof methods. Moreover,
the user can use locales, which allow defining local scopes in
which constants are defined and assumptions about them are made. Theorems
can be proven in the context of a locale and can use the constants and
depend on the assumptions of this locale. A locale can be instantiated to
concrete entities if the user is able to show that the latter fulfill the
locale assumptions.

\paragraph{VCC}~\cite{cohen-etal:2009:vcc} is an assertional, automatic, deductive code
verifier for full C.  Specifications in the form of function
contracts, data invariants, and loop invariants as well as further
annotations to maintain inductively defined information or to guide
VCC otherwise, are added directly into the C source code as comments. 
During builds with a C compiler, these annotations are ignored.  
From the annotated program, VCC generates verification conditions for partial 
or total correctness, which it then tries to discharge using the automatic theorem 
prover Z3~\cite{moura-bjorner:2008:z3} through the Boogie 
verifier~\cite{barnett-etal:2006:boogie}.

Verification in VCC makes heavy use of ghost data and code, enclosed by
\inlinevcc{_(} and \inlinevcc{)}, for reasoning about the program but
omitted from the concrete implementation. VCC provides ghost objects, ghost
fields of structured data types, local ghost variables, ghost function parameters, 
and ghost code. Ghost data and ghost code can use both C data types and 
additional mathematical data types, \eg mathematical integers
(\inlinevcc{\\integer}) and natural numbers (\inlinevcc{\\natural}),
records (similar to C structures), and maps (with a syntax similar to C
arrays). VCC ensures that information does not flow from a ghost state to
a non-ghost state and that all ghost code terminates; these checks guarantee
that program execution, when projected to the non-ghost code, is not affected by
the ghost code.

\paragraph{Export from VCC to Isabelle/HOL.} For the types and propositions that we pass from 
VCC to Isabelle/HOL, we restrict ourselves to a subset of VCC's specification language.  
Simple types are natural numbers,
integers, algebraic datatypes over simple types, and ghost records whose
fields are simple types. Rich types are simple types, ghost records whose
fields are rich types, and maps from simple types to rich types.
Propositions can be formed by usual logical connectives, quantifiers over
variables of rich types, arithmetic expressions, equalities, and
user-defined pure, stateless functions whose argument and result types are
rich and whose definitions or contracts are again propositions, possibly
using pattern matching over algebraic datatypes. Any type or function of this subset can
be expressed equivalently in Isabelle/HOL, essentially by syntactic
rewriting. More precisely, VCC algebraic datatypes can be translated into
Isabelle datatypes, VCC ghost records can be translated into
Isabelle records, and pure VCC ghost functions can be translated into
Isabelle function definitions.  The former two translations are sound and
complete because the semantics of datatypes and records is the same in both
systems; the latter is sound and complete because VCC's underlying logic is
subsumed by the higher-order logic of Isabelle/HOL.  The translation maps
VCC specification types (\inlinevcc{\\bool}, \inlinevcc{\\natural},
\inlinevcc{\\integer}, and map types) to equivalent Isabelle types
($\mathit{bool}$, $\mathit{nat}$, $\mathit{int}$, and function types) and
maps VCC expressions comprising logical connectives, quantifiers,
arithmetic operations, equality, and specification functions to corresponding 
Isabelle terms.

\section{Case Studies}
  \label{sec:case-studies}

We present three case studies from the domain of graphs. We start with 
the certifying algorithms and the corresponding checkers in LEDA.  We give 
formal proofs for the correctness of the checkers, for the related witness properties, 
and for the connection between them.  Except for the witness properties,
which are proven
in Isabelle/HOL, all presented abstractions and functions have been
formally verified using VCC.

All files related to our formalization
can be obtained from the following URL: \\
\url{http://www21.in.tum.de/~boehmes/certifying\_computations.html}

\subsection{Connected Components in Graphs}
  \label{sec:case-study-connected}
Our first case study considers the connected components problem. 
Given an undirected graph $G = (V, E)$, we consider an algorithm that
decides whether $G$ is connected, i.e., whether there is a path between any pair of
vertices~\cite[Section 7.4]{melhorn-naeher:1999:leda}.  In the negative
case, \ie when the graph is not connected, there is a simple witness.
It consists of a cut $S$, \ie a nonempty subset $S$ of the vertices with $S
\neq V$ such that every edge of the graph has either both or no endpoint
in $S$.  In other words, no edge crosses the cut. In the positive case,
\ie when the given graph is connected, the algorithm can produce 
a spanning tree of $G$ as a witness. A spanning tree of $G$ is a subgraph of
$G$, which is a tree and contains all vertices of $G$. We concentrate here
on the more complicated positive case.  We describe a checker for the
spanning tree witness and the verification of this checker.  On a high
level, we instantiate our general approach as follows:

\begin{center}
\begin{tabular}{r@{\hspace{1em}=\hspace{1em}}p{9cm}}
input $x$ & an undirected graph $G = (V, E)$\\
output $y$ & either $\mathit{True}$ or $\mathit{False}$, indicating whether $G$ is
connected \\
witness $w$ & a cut or a spanning tree \\
$\Pre(x)$ & $G$ is wellformed, \ie $E \subseteq V \times V$, $V$ and $E$ 
are finite sets\\
$\Wit(x, y, w)$ & $y$ is $\mathit{True}$ and $w$ is a spanning tree of $G$, 
or $y$ is $\mathit{False}$ and $w$ is a cut \\
$\Post(x,y)$ & if $y$ is $\mathit{True}$, $G$ is connected, and if $y$ is
$\mathit{False}$, $G$ is not connected.
\end{tabular}
\end{center}

\newcommand{\parentedge}{\mathit{parent\dash{}edge}} 
\newcommand{\num}{\mathit{num}}

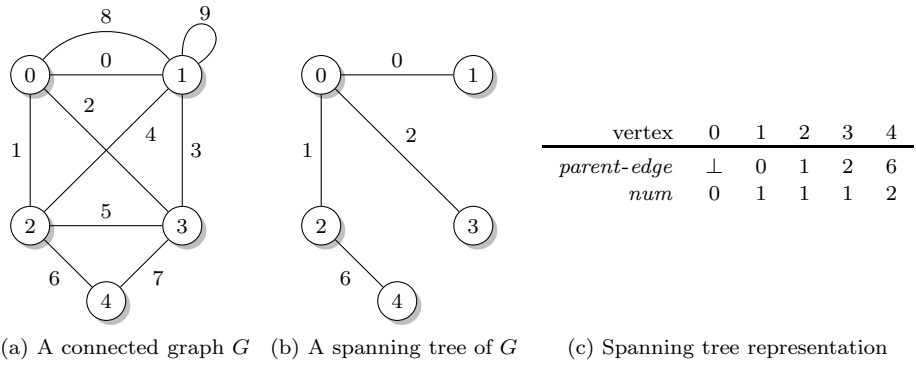
\begin{figure}
\begin{subfigure}[b]{0.28\textwidth}
  \centering
	\begin{tikzpicture}
    \tikzstyle{circ}=[circle, draw=black, fill=white, drop shadow]
    \node[circ] (n0) at (0,2) {0};
    \node[circ] (n1) at (2,2) {1};
    \node[circ] (n2) at (0,0) {2};
    \node[circ] (n3) at (2,0) {3};
    \node[circ] (n4) at (1,-1) {4};
    \draw (n0) to node[above] {0} (n1);
    \draw (n0) to node[left] {1} (n2);
    \draw (n0) to node[above right, near start] {2} (n3);
    \draw (n1) to node[right] {3} (n3);
    \draw (n1) to node[below right, near start] {4} (n2);
    \draw (n2) to node[above] {5} (n3);
    \draw (n2) to node[below left] {6} (n4);
    \draw (n3) to node[below right] {7} (n4);
    \draw (n0) to[bend left=50] node[above] {8} (n1);
    \draw (n1) to[in=90, out=35, loop] node[above right, near end] {9} (n1);
  \end{tikzpicture}
  \caption{A connected graph $G$}
  \label{fig:connected-components-a}
\end{subfigure}
\hfill
\begin{subfigure}[b]{0.28\textwidth}
  \centering
  \begin{tikzpicture}
    \tikzstyle{circ}=[circle, draw=black, fill=white, drop shadow]
    \node[circ] (n0) at (0,2) {0};
    \node[circ] (n1) at (2,2) {1};
    \node[circ] (n2) at (0,0) {2};
    \node[circ] (n3) at (2,0) {3};
    \node[circ] (n4) at (1,-1) {4};
    \draw (n0) to node[above] {0} (n1);
    \draw (n0) to node[left] {1} (n2);
    \draw (n0) to node[above right] {2} (n3);
    \draw (n2) to node[below left] {6} (n4);
  \end{tikzpicture}
  \caption{A spanning tree of $G$}
  \label{fig:connected-components-b}
\end{subfigure}
\hfill
\begin{subfigure}[b]{0.42\textwidth}
  \centering
  \begin{tabular}{rccccc}
  vertex & 0 & 1 & 2 & 3 & 4 \\
  \hline\\[-0.7em]
  $\parentedge$ & $\bot$ & 0 & 1 & 2 & 6 \\[0.3ex]
  $\num$ & 0 & 1 & 1 & 1 & 2
  \end{tabular}\\[1.5cm]
  \caption{Spanning tree representation}
  \label{fig:connected-components-c}
\end{subfigure}

\caption{An example of a connected graph $G$ and a spanning tree of $G$
  witnessing its connectivity.  The vertices belong to the set $\{0,\dots,
  n-1\}$ and the edges are pairs of vertices indexed by an identifier
  ranging from $0$ to $m-1$, where $n$ and $m$ are the number of vertices
  and edges in $G$. The spanning tree in (b) can be represented by a root
  vertex $r=0$ and functions $\parentedge$ and $\num$ as shown in the table
  in (c). Graphs may have self-loops and parallel edges.}
  \label{fig:connected-components}
\end{figure}

\noindent
We restrict ourselves to the positive case $y = \mathit{True}$.
Figure~\ref{fig:connected-components} shows a graph $G$ and a spanning tree
of it. We represent spanning trees by functions $\parentedge$ and $\num$ 
and by a root vertex $r$, and we view the edges of the tree oriented towards $r$:
for $v \not= r$, $\parentedge(v)$ is the first edge on the path from $v$ to
$r$, $\parentedge(r) = \bot$, and $\num(v)$ is the length of the path from
$v$ to $r$ for all $v$.  The function $\num$ is needed in order to show that 
$\parentedge$ encodes a forest. We present the implementation of a checker in
Section~\ref{sec:connected-checker}, detail the formalization of the
witness predicate and the verification of the checker in
Section~\ref{sec:connected-vcc}, and prove the witness property in
Section~\ref{sec:connected-isabelle}. The relevant files in the companion
website are check\_connected.c (C-code and checker correctness),
check\_connected.thy (Isabelle/HOL representation of the VCC checker
specification and witness property), ConnectedComponents.thy (abstract
verification of the witness property in a locale), and
ConnectedComponents\_Link.thy (instantiation of the locale with the
Isabelle/HOL representation of the VCC specification).

\subsubsection{Connected-Components Checker}
  \label{sec:connected-checker}

We begin by fixing a representation of graphs in the programming
language C, see Listing~\ref{lst:c-graph}. Vertices are numbered 
consecutively from $0$ to $\text{\inlinevcc{n}} - 1$.  Edges are pairs where the first
vertex is labeled \inlinevcc{s} (for source), and the second vertex is
labeled \inlinevcc{t} (for target).  Edges are stored in an array
\inlinevcc{es}, which is indexed by edge identifiers ranging from 0 to
$\text{\inlinevcc{m}} - 1$.  We require that the two vertices
of each edge belong to the graph, \ie that they are from the range $\{0,
\ldots, \text{\inlinevcc{n}} - 1\}$, and call graphs with this property \emph{wellformed}. 
We use the same data structure for directed and undirected graphs. For
directed graphs, an edge \inlinevcc{e} with \inlinevcc{e.s} = \inlinevcc{u}
and \inlinevcc{e.t} = \inlinevcc{v} is directed from \inlinevcc{u} to
\inlinevcc{v}. For undirected graphs, it represents the unordered pair
\{\inlinevcc{u},\inlinevcc{v}\}.

\begin{listing}
\begin{vcclst}
typedef unsigned Nat;
typedef Nat Vertex;
typedef Nat Edge_Id;
typedef struct { Nat s; Nat t; } Edge;
typedef struct { Nat m; Nat n; Edge* es; } Graph;
\end{vcclst}
\caption{A representation of graphs in C. The field \inlinevcc{m} gives the
  number of edges (and hence the length of the array \inlinevcc{es}), and
  \inlinevcc{n} gives the maximum number of vertices in the graph.}
  \label{lst:c-graph}
\end{listing}

We represent spanning trees as explained before. Instead of
functions, we
use two arrays, \inlinevcc{parent_edge} and \inlinevcc{num}, in addition to a
root vertex \inlinevcc{r}. The \inlinevcc{parent_edge} array maps
\inlinevcc{r} to a negative value, \ie to a value that does not identify
any edge.

The connected-graph checker is a function that succeeds if the two
functions \inlinevcc{check_r} and \inlinevcc{check_parent_num}
(Listing~\ref{lst:connected-checker}) succeed.  The first function checks
that \inlinevcc{r} is indeed the root of the spanning tree.  The second
function checks for every vertex \inlinevcc{v} different from \inlinevcc{r}
that the edge \inlinevcc{parent_edge[v]} is incident to \inlinevcc{v} and
that the other endpoint of the edge has a number one smaller than
\inlinevcc{num[v]}.

\begin{listing}
\begin{vcclst}
int check_r(Graph* G, Vertex r, int* parent_edge, Nat* num)
{
    return r < G->n && num[r] == 0 && parent_edge[r] == -1;
}

int check_parent_num(Graph* G, Vertex r, int* parent_edge, Nat* num)
{
    Vertex v, a, b; Edge_Id e; 
 
    for (v = 0; v < G->n; v++)
    {
        if (v == r) continue;
  
        if (parent_edge[v] < 0 || ((Edge_Id)parent_edge[v]) >= G->m) return FALSE;
  
        e = (Edge_Id)parent_edge[v];
        a = G->es[e].s;
        b = G->es[e].t;

        if (v == a && num[v] == num[b] + 1) continue;
        if (v == b && num[v] == num[a] + 1) continue;
        return FALSE;
    }
    return TRUE;
}
\end{vcclst}
\caption{The connected-components checker.}
  \label{lst:connected-checker}
\end{listing}

\subsubsection{Checker Correctness}
  \label{sec:connected-vcc}

\begin{listing}
\begin{minipage}[t]{.42\textwidth}
\begin{vcclst}
_(typedef \natural \Vertex)
_(typedef \natural \Edge_Id)
_(record \Edge {
    \Vertex src;
    \Vertex trg;
})
_(record \Graph {
    \natural num_verts;
    \natural num_edges;
    \Edge edge[\Edge_Id];
})
\end{vcclst}
\end{minipage}
\hfill
\begin{minipage}[t]{.55\textwidth}
\begin{vcclst}
_(def \bool \wellformed(\Graph G)
{
    return
        forall \Edge_Id i; i < G.num_edges ==>
            G.edge[i].src < G.num_verts &&&
            G.edge[i].trg < G.num_verts;
})
\end{vcclst}
\end{minipage}
\caption{Abstract graphs and a predicate to describe wellformed graphs.}
  \label{lst:spec-graph}
\end{listing}

To prove the two checker functions correct, we  need to provide abstract representations for graphs and paths. We decided to keep them close to the concrete representation for two
reasons. First, it makes detecting differences, and hence potential bugs, easier
for the programmer. Second, it also makes reasoning for VCC simpler.  The
declaration of abstract graphs is given in Listing~\ref{lst:spec-graph}
together with the ghost predicate \inlinevcc{\\wellformed} for describing
when an abstract graph is wellformed.  This ghost predicate plays the role
of the precondition $\Pre$ in this case study.
Our abstract version of the \inlinevcc{num} array is a mapping from
vertices to natural numbers. The abstract version of the
\inlinevcc{parent_edge} array is a mapping from vertices to the set $\Nat
\cup \{\bot\}$; we use $\bot$ to model an undefined value. To represent
this set, we define an algebraic datatype \inlinevcc{Option}:

\vspace{5pt} 

\begin{vccquote}
_(datatype \Option
{
    case \none();
    case \some(\Edge_Id e);
})
\end{vccquote}

\noindent
with operations \inlinevcc{\\is_some(o)} for the test $\text{\inlinevcc{o}}
\neq \bot$ and \inlinevcc{\\the(o)} for extracting an edge identifier.
The abstraction functions that map concrete data to pure mathematical data
are straightforward to define. For example,

\begin{vccquote}
_(def \Graph \abs_graph(Graph* G)
{
    return (\Graph) {
        .num_verts = G->n,
        .num_edges = G->m,
        .edge =
            \lambda \Edge_Id i;
                (i < G->m) ?
                    (\Edge) { .src = G->es[i].s, .trg = G->es[i].t } :
                    (\Edge) { .src = 0, .trg = 0 }};
})
\end{vccquote}

\noindent
abstracts a concrete graph $G$ into an abstract graph of type
\inlinevcc{\\Graph}.
Using the abstract types, we define the witness predicate as a conjunction
of two properties, one for each of the checker functions in
Listing~\ref{lst:connected-checker}.

\begin{description}
\item[\textsf{check\_r}:] Vertex \inlinevcc{r} is the root of the spanning
  tree:

\begin{vccquote}
r < G.num_verts &&& !!\is_some(parent_edge[r]) &&& num[r] = 0
\end{vccquote}

\item[\textsf{check\_parent\_num}:] Every vertex of the graph is connected
  to some other vertex closer to \inlinevcc{r}:

\begin{vccquote}
forall \Vertex v; v < G.num_verts &&& v !!= r ==>
    \is_some(parent_edge[v]) &&& \the(parent_edge[v]) < G.num_edges &&&
    (G.edge[\the(parent_edge[v])].trg == v &&&
       num[v] == num[G.edge[\the(parent_edge[v])].src] + 1 |||
     G.edge[\the(parent_edge[v])].src == v &&&
       num[v] == num[G.edge[\the(parent_edge[v])].trg] + 1)
\end{vccquote}

\end{description}

Thanks to the low level of abstraction in the above predicates, the two
checker functions are easily verified. For the verification of
\inlinevcc{check_parent_num}, we need to annotate the loop with the second
conjunct above in which \inlinevcc{G.num_verts} is replaced by the loop
variable as a loop invariant. Moreover, for every \inlinevcc{return
FALSE}, we need to assert, or restate, on the abstract level the properties
that are violated to guide VCC.  Otherwise, it would fail to show
completeness of the checker. For instance,

\begin{vccquote}
if (parent_edge[v] < 0 || ((Edge_Id)parent_edge[v]) >= G->m)
{
    _(assert !!\is_some(P[v]) ||| \the(P[v]) >== \abs_graph(G).num_edges)
    return FALSE;
}
\end{vccquote}

\noindent
is one of the two occurrences of such extra assertions in
\inlinevcc{check_parent_num}.

We express the postcondition of the checker, \ie that any pair of vertices
of the graph \inlinevcc{G} is connected by a path as follows:

\begin{vccquote}
forall \Vertex u, v; u < G.num_verts &&& v < G.num_verts ==>
    exists \Path p; \natural n; \is_path(G, p, n, u, v)
\end{vccquote}

\noindent
Here, the type \inlinevcc{\\Path} is a sequence of vertices, represented as
a mapping from natural numbers to vertices, and the predicate
\inlinevcc{\\is_path(G, p, n, u, v)} holds if the path \inlinevcc{p}
of length \inlinevcc{n} starts at \inlinevcc{u}, ends at \inlinevcc{v}, and only
contains pairwise distinct vertices  that are connected by edges of the
graph:

\begin{vccquote}
p[0] == u &&&
p[n] == v &&&
(forall \natural i; i <== n ==> p[i] < G.num_verts) &&&
(forall \natural i; i < n ==> \is_edge(G, p[i], p[i+1])) &&&
(forall \natural i, j; i <== n &&& j <== n &&& i !!= j ==> p[i] !!= p[j])
\end{vccquote}

\noindent
The predicate \inlinevcc{\\is_edge(G, u, v)}, for any two vertices
\inlinevcc{u} and \inlinevcc{v} of \inlinevcc{G}, is true if and only if
\inlinevcc{u} and \inlinevcc{v} are the endpoints of an edge of
\inlinevcc{G}:

\begin{vccquote}
exists \Edge_Id i; i < G.num_edges &&&
    (G.edge[i].src = u &&& G.edge[i].trg = v |||
     G.edge[i].src = v &&& G.edge[i].trg = u)
\end{vccquote}
\noindent
We give a formal proof for the implication, from precondition and witness
predicate to the postcondition, in the next section.

\subsubsection{Proof of the Witness Property for the Connected-Components
  Checker}
  \label{sec:connected-isabelle}
  
We prove in Isabelle that a spanning tree witnesses the connectivity of a
graph.  We do so in two steps.  First, we perform a high-level proof in which
we abstract from concrete representations of graphs and spanning trees.  We
then instantiate this proof with the data structures and corresponding
properties exported from VCC.

Our formalization builds on the Isabelle graph library developed by Lars
No\-schinski~\cite{noschinski:2011:graph-library}. Graphs in this library
are directed.  A \emph{pseudo-digraph} is a wellformed directed graph with
a finite set of vertices and a finite set of edges;  the library reserves
the word \emph{digraph} for graphs without parallel edges and self-loops.
Unlike in VCC, in Isabelle we represent undirected graphs as bidirected graphs\footnote{We do so in order to directly use the Isabelle graph library.}, 
i.e., directed graphs
containing for every edge $(u,v)$ also the reversed edge $(v,u)$.  The
function $\mathit{mk\dash symmetric}$ maps a pseudo-digraph to a bidirected
pseudo-digraph by appropriately extending the set of edges with missing
reversed edges.  A vertex $v$ is \emph{reachable} from a vertex $u$ in
a (bi)directed graph $G$ if there exists a
directed \emph{walk} from $u$ to $v$ in $G$, {\ie a sequence
$(u_1,v_1),(u_2,v_2),\ldots, (u_k,v_k)$ of edges with $u_1 = u$, $v_k = v$,
and $v_i = u_{i+1}$ for $1 \le i < k$. An alternative and
equivalent formalization of reachability between vertices $u$ and $v$ in
$G$ is via sequences of vertices $v_1, v_2, \ldots, v_k$, where $v_1 = u$,
$v_k = v$ and $(v_i, v_{i+1})$ is an edge of $G$ for $1 \le i < k$.  We say a 
vertex $v$ is \emph{reachable through a path} from a vertex $u$ in $G$ if $v$ 
is reachable from $u$ through a path in $G$. An undirected graph is \emph{connected} 
if for any two vertices of the graph, one is reachable through a path from the other.

\begin{listing}
\begin{isalst}
locale connected_components_locale = pseudo_digraph +
    fixes num :: "'a \<Rightarrow> nat"
    fixes parent_edge :: "'a \<Rightarrow> 'b option"
    fixes \r :: 'a
    assumes r_assms: "\r \<in> verts \G \<and> parent_edge \r = None \<and> num \r = \0"
    assumes parent_num_assms: 
        "\<And>\v. \v \<in> verts \G \<and> \v \<noteq> \r \<Longrightarrow>
            \<exists>\e \<in> edges \G. 
                parent_edge \v = Some \e \<and> 
                target \G \e = \v \<and> 
                num \v = num (start \G \e) + \1"
\end{isalst}
\caption{Locale for the connected-components proof in Isabelle. The symbol 
$\bigwedge$ stands for universal quantification.}
\label{lst:connected-components-locale}
\end{listing}

Our high-level proof rests on the Isabelle locale
\inlineisa{connected_components_locale}
(Listing~\ref{lst:connected-components-locale}) that describes the
assumptions of our theorem.  We fix $G$ to be a pseudo-digraph where
\inlineisa{'a} is an abstraction of the type of vertices and \inlineisa{'b}
is an abstraction of the type of edges.  Furthermore, we fix a 
representation of spanning trees with functions \inlineisa{parent_edge} and
\inlineisa{num} and vertex $r$ as the root. Based on these assumptions, we
prove that $G$ is connected.  We show first that every vertex $v$ in the
graph is reachable from the root $r$ by induction on $\mathit{num}\ v$, \ie
the length of the walk from $r$ to $v$ in the spanning tree.  The base case
follows directly from our assumptions.  For the inductive step, we can
assume a walk from $r$ to the parent of a vertex $v$.  Using the
assumptions, this walk can be extended to a walk from $r$ to $v$ since
there is an edge between $v$ and its parent.  Now, since $G$ is bidirected,
we can establish that there is a walk between any two vertices of $G$ by
combining the walks that connect them with the root $r$. If there is a walk 
between two vertices, there is also a path between them. Therefore, all vertices 
in $G$ are reachable through a path from one another, and hence, $G$ is connected.

The final part of the formal proof---linking the
high-level proofs with the properties exported from VCC to Isabelle---is
fairly straightforward. Proving that the precondition
and the witness predicate (\cf Section~\ref{sec:connected-vcc}) match the
assumptions specified in the locale \inlineisa{connected_components_locale}
involves no reasoning beyond syntactical rewriting.  To instantiate these assumptions, 
we provide lifting functions that abstract from the concrete representations 
of graphs and spanning trees stemming from our VCC specification to the 
high-level representation used by the Isabelle graph library.  Thus, if the 
checker accepts, the lifted high-level graph is connected. Establishing the 
checker postcondition (the connectivity of unlifted graphs) requires showing 
that any high-level path witnessing reachability between two vertices 
corresponds to an unlifted path. This is straightforward  because our 
representation of paths in the VCC formalization (\cf Section~\ref{sec:connected-vcc}) 
is close to the path representation of the Isabelle graph library.

\subsection{Shortest Paths in Graphs}
  \label{sec:case-study-shortest}

\newcommand{\dist}{\mathit{dist}}
\newcommand{\cost}{\mathit{cost}}

Our second case study is about the single-source shortest-path problem in 
directed graphs with nonnegative edge weights. It can be solved, for instance, 
using Dijkstra's algorithm~\cite[Sections 6.6 and 7.5]{melhorn-naeher:1999:leda}. 
Instead of verifying this algorithm, we request that it returns, along with the
computed shortest distance for each vertex of a graph, the corresponding
shortest path as witness.  That is, we instantiate our general approach as
follows:

\begin{center}
\begin{tabular}{r@{\hspace{1em}=\hspace{1em}}p{9cm}}
input $x$ & a directed graph $G = (V, E)$, a function
$\cost : E \rightarrow \Nat$ for edge weights, a vertex $s$ \\
output $y$ & a mapping $\dist : V \rightarrow (\Nat \cup \infty)$\\
witness $w$ & a tree rooted at $s$\\
$\Pre(x)$ & $G$ is wellformed, V and E are finite sets, and $s \in V$\\
$\Wit(x, y, w)$ & $w$ is a shortest-path tree, i.e., for each $v$ reachable
from $s$, the tree path from $s$ to $v$ has length $\dist(v)$\\
$\Post(x,y)$ & for
each $v \in V$, $\dist(v)$ is the cost of a shortest path from $s$ to
$v$ (or $\infty$ if there is no path from $s$ to $v$).
\end{tabular}
\end{center}

\begin{figure}
\begin{subfigure}[b]{0.25\textwidth}
  \centering
	\begin{tikzpicture}
    \tikzstyle{circ}=[circle, draw=black, fill=white, drop shadow]
    \node[circ] (n0) at (0,2) {$s$};
    \node[circ] (n1) at (1.5,2) {$t$};
    \node[circ] (n2) at (0,0) {$u$};
    \node[circ] (n3) at (1.5,0) {$v$};
    \node[circ] (n4) at (0.75,-1) {$w$};
    \draw[-latex] (n0) to node[above] {0/\bfseries 1} (n1);
    \draw[-latex] (n0) to node[left] {1/\bfseries 1} (n2);
    \draw[-latex] (n1) to node[above left, pos=0.4] {2/\bfseries 1} (n2);
    \draw[-latex] (n1) to node[left, near end] {3/\bfseries 0} (n3);
    \draw[-latex] (n3) to[bend right=60] node[left] {4/\bfseries 0} (n1);
    \draw[-latex] (n4) to[in=115, out=55, loop] node[above] {5/\bfseries 2}
      (n4);
  \end{tikzpicture}
  \caption{A directed graph $G$}
  \label{fig:shortest-paths-a}
\end{subfigure}
\hfill
\begin{subfigure}[b]{0.31\textwidth}
  \centering
	\begin{tikzpicture}
    \tikzstyle{circ}=[circle, draw=black, fill=white, drop shadow]
    \node[circ] (n0) at (0,2) {$s$};
    \node[circ] (n1) at (1.5,2) {$t$};
    \node[circ] (n2) at (0,0) {$u$};
    \node[circ] (n3) at (1.5,0) {$v$};
    \node[white] (n4) at (0.75,-1) {$w$};
    \draw[-latex] (n0) to node[above] {0} (n1);
    \draw[-latex] (n0) to node[left] {1} (n2);
    \draw[-latex] (n1) to node[left] {3} (n3);
  \end{tikzpicture}
  \caption{A shortest-path tree of $G$}
  \label{fig:shortest-paths-b}
\end{subfigure}
\hfill
\begin{subfigure}[b]{0.42\textwidth}
  \centering
  \begin{tabular}{rccccc}
  vertex & $s$ & $t$ & $u$ & $v$ & $w$ \\
  \hline\\[-0.7em]
  $\parentedge$ & $\bot$ & 0 & 1 & 3 & $\bot$ \\[0.3ex]
  $\num$ & 0 & 1 & 1 & 2 & $\infty$ \\[0.3ex]
  $\dist$ & 0 & 1 & 1 & 1 & $\infty$
  \end{tabular}\\[1.2cm]
  \caption{Tree representation}
  \label{fig:shortest-paths-c}
\end{subfigure}
\caption{A directed graph $G=(V,E)$ with the edges labeled $i/\mathbf{k}$,
  where $i$ is a unique edge index and where $\mathbf{k}$ is the cost of
  that edge, and a shortest-path tree of $G$ rooted at start vertex $s\in
  V$.  The tree is encoded by $\parentedge$, $\num$ and $\dist$ according
  to the table in (c). Observe that vertex $w$ is not reachable from $s$ and
  that the cycle $t \rightarrow v \rightarrow t$ has cost zero.}
  \label{fig:shortest-paths}
\end{figure}
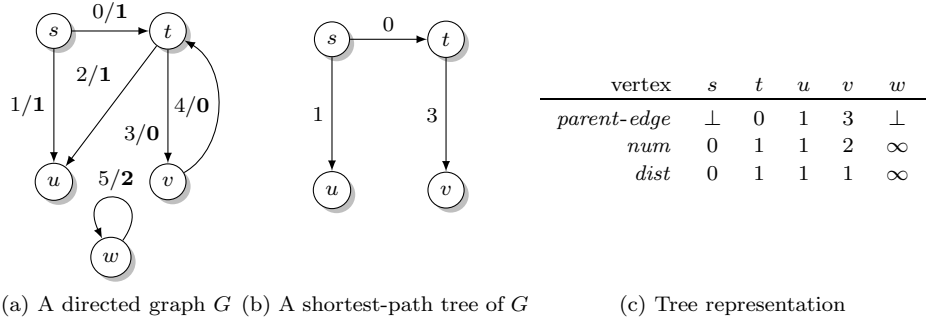

\noindent
Figure~{\ref{fig:shortest-paths} shows a directed graph and a shortest-path
tree rooted at $s$. We encode a shortest-path tree by functions
$\parentedge$, $\dist$, and $\num$. For each $v$ reachable from $s$,
$\dist(v)$ is the shortest-path distance from $s$ to $v$ and $\num(v)$ is
the depth of $v$ in the shortest-path tree. For vertices $v$ that are not
reachable from $s$, $\dist(v) = \num(v) = \infty$. For reachable vertices
$v$ different from $s$, the edge $\parentedge(v)$ is the last edge on a
shortest path from $s$ to $v$. This witness is somewhat verbose. As we will
see in Section~\ref{sec:shortest-isabelle}, we could do without the
$\parentedge$ function. If all edge costs are positive, no witness is
required beyond the $\dist$ function. If one also allows cost
zero for edges as we do, the depth function $\num$ is indispensable~\cite[Section 2.4]{mcconnell-etal:2011:certifying}.

We present the implementation of a checker in 
Section~\ref{sec:shortest-checker}, detail the formalization
of the witness predicate and the verification of the checker in 
Section~\ref{sec:shortest-vcc}, and prove the witness property in 
Section~\ref{sec:shortest-isabelle}. The relevant files in the companion
website are check\_shortest\_path.c (C-code and checker correctness),
check\_shortest\_path.thy (Isabelle/HOL representation of the VCC checker
specification, and witness property), ShortestPath.thy (abstract
verification of the witness property in a locale), and
ShortestPath\_Link.thy (instantiation of the locale with the Isabelle/HOL
representation of the VCC specification).

\subsubsection{Shortest-Paths Checker}
  \label{sec:shortest-checker}

We adopt the data structures of the previous case study
(Section~\ref{sec:connected-checker}) with the exception that the
\inlinevcc{num} array stores elements of type \inlinevcc{int} instead of
\inlinevcc{Nat}.  This is because vertices may now also be unreachable 
from the source vertex, and we encode this by
requiring that \inlinevcc{num} takes a negative value for such vertices.
We represent distances from the source vertex to any other vertex by an
array \inlinevcc{dist} with elements of type \inlinevcc{int}.  Any negative
value encodes $\infty$.  Finally, the edge weights are modeled by an array
\inlinevcc{cost} that gives for every edge a value of type
\inlinevcc{ushort} (an abbreviation for \inlinevcc{unsigned short}).

Based on these types, we implement the shortest-path checker as a function
that succeeds when all of the four functions given in
Listing~\ref{lst:shortest-checker} succeed.  That is, we check that the
source vertex \inlinevcc{s} is indeed the starting point (in
\inlinevcc{check_start_val}), that the \inlinevcc{dist} and \inlinevcc{num}
arrays are consistent with respect to unreachable vertices, \ie either both
are finite or both are infinite (in \inlinevcc{check_no_path}), that the
triangle inequality property (Section~\ref{sec:shortest-isabelle}) is
fulfilled (in \inlinevcc{check_trian}), and that the parent edge of every
vertex \inlinevcc{v} defines its distance value (in
\inlinevcc{check_just}).

There is a subtle point in the checker code. We want to establish the
triangle inequality ($\dist(u) + cost(u,v) \ge \dist(v)$ for all edges
$(u,v)$) and the distance justification ($\dist(u) + cost(u,v) = \dist(v)$
if $(u,v)$ is the parent edge of $v$) over the extended
natural numbers $\Nat \cup \{\infty\}$. However, C knows
only finite precision arithmetic. We solve the case of infinite distances
by appropriate case distinctions. We solve the case of potential overflow
in finite precision arithmetic as follows: Distances are of type
\inlinevcc{int}, \ie from the set $\{-2^{31}, \ldots, 2^{31} - 1\}$
on a 32-bit platform, and edge costs are of type
\inlinevcc{ushort}, \ie between $0$ and $2^{16} - 1$, and hence
contained in the set of nonnegative values of type \inlinevcc{int}.  In
arithmetic expressions, we cast all nonnegative values to
\inlinevcc{unsigned} with range $0 \ldots 2^{32} - 1$.  This guarantees
that finite precision arithmetic is exact and allows VCC to conclude
equalities and inequalities between natural numbers.

%
Note that there is an alternative approach where \inlinevcc{parent_edge} is not part of the 
witness. In that case \inlinevcc{check_just} has to be rewritten.
  When considering a node
\inlinevcc{v}, it has to iterate over all edges into \inlinevcc{v} to find
the edge that defines \inlinevcc{dist[v]}. An efficient implementation of
this iteration requires providing each vertex with the list of edges into
it.

\begin{listing}
\begin{vcclst}
bool check_start_val(Vertex s, int* dist, int* num)
{
    return dist[s] == 0;
}

bool check_no_path(Graph* G, int* dist, int* num)
{
    Vertex v;

    for (v = 0; v < G->n; v++)
    {
        if (INF(dist[v]) != INF(num[v])) return FALSE;
    }
    return TRUE;
}

int check_trian(Graph* G, ushort* cost, int* dist)
{
    Edge_Id e; Vertex source, target;

    for (e = 0; e < G->m; e++)
    {
        source = G->es[e].s;
        target = G->es[e].t;

        if (INF(dist[source])) continue;
        if (INF(dist[target])) return FALSE;
        if (VAL(dist[target]) > VAL(dist[source]) + cost[e]) return FALSE;
    }
    return TRUE;
}

bool check_just(Graph* G, Vertex s, ushort* cost, int* dist, int* parent_edge, int* num)
{
    Vertex v, source; Edge_Id e;

    for (v = 0; v < G->n; v++)
    {
        if (v == s || INF(num[v])) continue;
        if (parent_edge[v] < 0 || ((Edge_Id)parent_edge[v]) >= G->m) return FALSE;

        e = (Edge_Id)parent_edge[v];
        source = G->es[e].s;

        if (G->es[e].t != v) return FALSE;
        if (INF(dist[source]) || VAL(dist[v]) != VAL(dist[source]) + cost[e]) return FALSE;
        if (INF(num[source]) || VAL(num[v]) != VAL(num[source]) + 1) return FALSE;
    }
    return TRUE;
}
\end{vcclst}
\caption{Functions composing the shortest-path checker. The
  predicate \inlinevcc{INF(x)} abbreviates \inlinevcc{x < 0}, and
  \inlinevcc{VAL(x)} stands for the type cast \inlinevcc{(Nat)x};
  \inlinevcc{Nat} is the C type \inlinevcc{unsigned} as defined in
  Listing~\ref{lst:c-graph}.}
  \label{lst:shortest-checker}
\end{listing}

\subsubsection{Checker Correctness}
  \label{sec:shortest-vcc}

We now define our abstract specification for the shortest-path checker.
We use the same data structures as in the previous case study (Section~\ref{sec:connected-vcc}) with the exception that the
\inlinevcc{num} mapping now takes vertices to extended naturals ($\Nat \cup \{\infty\}$), represented by the type \inlinevcc{Enat}. Extended naturals provide an explicit value for infinity:

\begin{vccquote}
_(datatype \Enat
{
    case \enat_inf();
    case \enat_val(\natural n);
})
\end{vccquote}

\noindent
We define functions \inlinevcc{\\is_enat_inf} to check whether an extended
natural is infinity and \inlinevcc{\\enat_val_of} to convert an extended
natural distinct from infinity into the corresponding natural number.  For better
readability, we will write \inlinevcc{a $\;=_e\;$ $\infty$}
for \inlinevcc{\\is_enat_inf(a)}.  Moreover, we provide the
predicates \inlinevcc{\\enat\_eq} (abbreviated by $=_e$) and
\inlinevcc{\\enat_le} ($\le_e$) to decide equality and
less-or-equal of two extended naturals as well as a function
\inlinevcc{\\enat_add} ($+_e$) for the sum
of an extended natural and a natural number:

\begin{vccquote}
_(def \bool \enat_eq(\Enat e1, \Enat e2)
{
    return
        (e1 $=_e$ $\infty$ &&& e2 $=_e$ $\infty$) |||
        (e1 $\neq_e$ $\infty$ &&& e2 $\neq_e$ $\infty$ &&& \enat_val_of(e1) = \enat_val_of(e2));
})

_(def \bool \enat_le(\Enat e1, \Enat e2)
{
    return e2 $=_e$ $\infty$ ||| (e1 $\neq_e$ $\infty$ &&& e2 $\neq_e$ $\infty$ &&& \enat_val_of(e1) <== \enat_val_of(e2));
})

_(def \Enat \enat_add(\Enat e, \natural n)
{
    return (e $=_e$ $\infty$) ? \enat_inf() : \enat_val(\enat_val_of(e) + n);
})
\end{vccquote}

\noindent
The type of extended natural numbers is also used for the abstract
representation of the \inlinevcc{dist} array.
Again, as in the previous case study, concrete types and abstract types are
sufficiently similar such that abstraction functions relating one to the
other are straightforward to define.  We omit them here.

The preconditions of this case study are that \inlinevcc{G} is a wellformed
graph and that the source vertex \inlinevcc{s} is a vertex of
\inlinevcc{G}, \ie that \inlinevcc{s < G.num_verts} holds.
We formalize the witness predicate as a conjunction of four properties, one
for each of the four checker functions in
Listing~\ref{lst:shortest-checker}.

\begin{description}
\item[\textsf{check\_start\_val}:]
  Vertex \inlinevcc{s} is indeed the starting point:

\begin{vccquote}
dist[s] $=_e$ \enat_val(0)
\end{vccquote}

\item[\textsf{check\_no\_path}:] The \inlinevcc{num} mapping and the \inlinevcc{dist}
  mapping are consistent with respect to unreachable vertices, \ie both
  are either finite or infinite:

\begin{vccquote}
forall \Vertex v; v < G.num_verts ==> (dist[v] $=_e$ $\infty$ <--> num[v] $=_e$ $\infty$)
\end{vccquote}

\item[\textsf{check\_trian}:] The triangle inequality holds for all edges of
  the graph:

\begin{vccquote}
forall \Edge_Id i; i < G.num_edges ==>
    dist[G.edge[i].trg] $\le_e$ dist[G.edge[i].src] $+_e$ cost[i]
\end{vccquote}

\item[\textsf{check\_just}:] The parent edges encode a tree rooted at
  \inlinevcc{s} and define the distance values of reachable vertices:

\begin{vccquote}
forall \Vertex v;
    v < G.num_verts &&& v !!= s &&& num[v] $\neq_e$ $\infty$ ==>
    \is_some(parent_edge[v]) &&& \the(parent_edge[v]) < G.num_edges &&&
    v = G.edge[\the(parent_edge[v])].trg &&&
    dist[v] $=_e$ dist[G.edge[\the(parent_edge[v])].src] $+_e$ cost[\the(parent_edge[v])] &&&
    num[v] $=_e$ num[G.edge[\the(parent_edge[v])].src] $+_e$ 1
\end{vccquote}
\end{description}

\noindent
We have verified that each of these four properties holds if and only if the
corresponding checker function succeeds.  The three functions
\inlinevcc{check_no_path}, \inlinevcc{check_trian}, and
\inlinevcc{check_just} need additional annotations before VCC can verify
their correctness.  The loops in these functions have to be annotated with
loop invariants that are, just as in the previous case study (Section~\ref{sec:connected-vcc}), only simple variants of the
postconditions above.  Also, as for the connected-components checker, we
need to explicitly state properties that are
violated before every \inlinevcc{return FALSE} statement.  Such properties
are reformulations of concrete properties on the abstract level.  In
addition, both \inlinevcc{check_trian} and \inlinevcc{check_just} require
 the graph under consideration to be wellformed, and
\inlinevcc{check_just}, furthermore, requires that \inlinevcc{num} and
\inlinevcc{dist} are consistent (the postcondition of
\inlinevcc{check_no_path}). We add these requirements as preconditions to
the checker functions.

In order to be able to express the postcondition of the shortest-path checker,
we define sequences of edges as a recursive datatype:

\begin{vccquote}
_(datatype \Path
{
    case none();
    case path(\Edge_Id i, \Path p);
})
\end{vccquote}

\noindent
Only particular instances of this datatype are paths in the given
graph~\inlinevcc{G}.  To qualify valid paths, we proceed in two steps. We
first define a predicate that expresses the conditions under which a
sequence of edges constitutes a walk in graph \inlinevcc{G} from vertex
\inlinevcc{u} to vertex \inlinevcc{v} (Listing~\ref{lst:vcc-is-walk}).
Second, we define a predicate to describe when the set of vertices of an
edge sequence is distinct (Listing~\ref{lst:vcc-distinct-verts}). A
path from vertex \inlinevcc{u} to vertex \inlinevcc{v} in \inlinevcc{G} is
a walk \inlinevcc{p} from \inlinevcc{u} to \inlinevcc{v} with distinct
vertices. We define this as a predicate \inlinevcc{\\is_path(G, p, u, v)}.

\begin{listing}
\begin{vcclst}
_(def \bool \is_walk(\Graph G, \Path p, \Vertex u, \Vertex v)
{
    switch (p)
    {
        case none(): return u = v;
        case path(i, q):
            return i < G.num_edges &&& u = G.edge[i].src &&& \is_walk(G, q, G.edge[i].trg, v);
    }
})
\end{vcclst}
\caption{A walk from vertex \inlinevcc{u} to vertex \inlinevcc{v} is a
  finite sequence of connected edges of graph \inlinevcc{G} where the
  source vertex of the first edge is \inlinevcc{u} and the target vertex of
  the last edge is \inlinevcc{v}.}
  \label{lst:vcc-is-walk}
\end{listing}

\begin{listing}
\begin{vcclst}
_(def \bool \occurs(\Graph G, \Vertex u, \Vertex v, \Path p)
    _(decreases \size(p))
{
    switch (p)
    {
        case none(): return u = v;
        case path(i, q): return u = G.edge[i].src ||| \occurs(G, u, G.edge[i].trg, q);
    }
})

_(def \bool \distinct_verts(\Graph G, \Path p)
{
    switch (p)
    {
        case none(): return \true;
        case path(i, q): return !!\occurs(G, G.edge[i].src, G.edge[i].trg, q) &&& \distinct_verts(G, q);
    }
})
\end{vcclst}
\caption{Predicate \inlinevcc{\\distinct\_verts(G, p)} holds if the set of
  vertices connected by path \inlinevcc{p} is distinct. Predicate
  \inlinevcc{\\occurs(G, u, v, p)} is true if and only if \inlinevcc{u} is
  either equal to \inlinevcc{v} or equal to any vertex touched by path
  \inlinevcc{p}.}
  \label{lst:vcc-distinct-verts}
\end{listing}

With a recursive function \inlinevcc{\\path\_cost} that computes for a
given path its length using the \inlinevcc{cost} mapping, we can finally
state the postcondition of the shortest path checker:

\begin{vccquote}
(forall \Vertex v; v < G.num_verts ==>
    !!\is_enat_inf(dist[v]) <--> (exists \Path p; \is_path(G, p, s, v))) &&&
(forall \Vertex v; v < G.num_verts &&& !!\is_enat_inf(dist[v]) ==>
    (forall \Path p; \is_path(G, p, s, v) ==> \enat_val_of(dist[v]) <== \path_cost(cost, p)) &&&
    (exists \Path p; \is_path(G, p, s, v) &&& \enat_val_of(dist[v]) = \path_cost(cost, p)))
\end{vccquote}

\noindent
We formally prove this property, under the assumption of
the precondition and the witness predicate, in Isabelle below.

\subsubsection{Proof of the Witness Property for the Shortest-Path Checker}
  \label{sec:shortest-isabelle}
  
We here present the outline of the Isabelle/HOL proof of the witness property,
namely that if $G$ is wellformed (the
precondition of this case study) and if the witness predicate holds, then 
${\mathit dist}(v)$ is indeed the shortest-path
distance from $s$ to $v$ for all vertices $v \in V$. 
For the full formal proof, please check the companion website.

\begin{listing}
\begin{isalst}
locale basic_sp = pseudo_digraph +
    fixes dist :: "'a \<Rightarrow> ereal"
    fixes \c :: "'b \<Rightarrow> real"
    fixes \s :: "'a"
    assumes general_source_val: "dist \s \<le> \0"
    assumes trian: "\<And>\e. \e \<in> edges \G \<Longrightarrow> dist (target \G \e) \<le> dist (start \G \e) + \c \e"

locale basic_just_sp = basic_sp +
    fixes num :: "'a \<Rightarrow> enat"
    assumes just:
      "\<And>\v. \v \<in> verts \G \<Longrightarrow> \v \<noteq> \s \<Longrightarrow> num \v \<noteq> \<infinity> \<Longrightarrow>
          \<exists> \e \<in> edges \G.
              \v = target \G \e \<and>
              dist \v = dist (start \G \e) + \c \e  \<and>
              num \v = num (start \G \e) + enat \1"

locale shortest_path_pos_cost = basic_just_sp +
    assumes s_in_G: "\s \<in> verts \G"
    assumes start_val: "dist \s = \0"
    assumes no_path: "\<And>\v. \v \<in> verts \G \<Longrightarrow> (dist \v = \<infinity> \<longleftrightarrow> num \v = \<infinity>)"
    assumes pos_cost: "\<And>\e. \e \<in> edges \G \<Longrightarrow> \0 \<le> \c \e"
\end{isalst}
\caption{Locales for the shortest-paths proof in Isabelle.}
\label{lst:shortest-path-locale-isabelle}
\end{listing}

Listing~\ref{lst:shortest-path-locale-isabelle} shows our Isabelle locales. 
We separate the assumptions into three locales to avoid the use of unneeded 
assumptions when proving intermediate lemmas. This makes the intermediate 
lemmas more general, and hence, usable in other contexts\footnote{For example, 
we used some of those lemmas also for the verification of a checker for the 
shortest-path problem with general edge weights (not only nonnegative edge 
weights as in this case study)~\cite{rizkallah:2012:shortest}.}. The first locale
\inlineisa{basic_sp} subsumes the locale \inlineisa{pseudo_digraph} mentioned 
in Section~\ref{sec:connected-isabelle}. Moreover, it assumes it is given 
the function $\dist: V\to (\mathbb{R}\cup\{\infty, -\infty\})$, an edge cost function 
$c: E \to \mathbb{R}$, and a start vertex $s$.

Let $\mu$ be a function that takes a cost function $c$ on edges and two 
vertices $s$ and $v$ and returns the cost of a shortest path from $s$ to
$v$ in $G$ using the cost function $c$. We split the proof
into two parts. First, 
we prove a lemma $\mathit{\dist\dash{}le\dash{}\mu}$ using the locale 
\inlineisa{basic_sp}. The lemma states that $\dist\ v\leq\mu\ c\ s\ v$ for every 
vertex $v\in V$. Then, we prove a lemma $\mathit{dist\dash{}ge\dash{}\mu}$ 
using the locale \inlineisa{basic_just_sp}. The lemma  states that $\dist\ v\geq\mu\ c\ s\ v$ 
for every vertex $v\in V$ under some extra assumptions 
(Listing~\ref{lst:shortest-path-dist-ge-mu-isabelle}). Later,
we show that these extra
assumptions hold in the locale \inlineisa{shortest_path_pos_cost}. Hence, we obtain 
a theorem stating that $\dist\ v= \mu\ c\ s\ v$ for every $v\in V$ using the locale 
\inlineisa{shortest_path_pos_cost}. 

\begin{listing}
\begin{isalst}
lemma (in basic_just_sp) dist_ge_\<mu>:
    fixes \v :: 'a
    assumes "\v \<in> verts \G"
    assumes "num \v \<noteq> \<infinity>"
    assumes "dist \v \<noteq> -\<infinity>"
    assumes "\<mu> \c \s \s = ereal \0"
    assumes "dist \s = \0"
    assumes "\<And>\u. \u \<in> verts \G \<Longrightarrow> \u \<noteq> \s \<Longrightarrow> num \u \<noteq> \<infinity> \<Longrightarrow> num \u \<noteq> enat \0"
    shows "dist \v \<ge> \<mu> \c \s \v"
\end{isalst}
\caption{The central lemma of the shortest-paths proof in Isabelle}
\label{lst:shortest-path-dist-ge-mu-isabelle}
\end{listing}

Linking this Isabelle proof with the specification exported from VCC is a
matter of translating from one representation to another. We intentionally 
chose to define paths and their costs in VCC (Section~\ref{sec:shortest-vcc}) 
similar to the way they are defined in the Isabelle graph library to ease our 
translation proofs. Since there are several more concepts to relate than in the 
previous checker (Section~\ref{sec:connected-isabelle}), our proofs for the shortest-path
checker are more tedious. Nevertheless, no complex reasoning is required.
We establish that the assumptions of the \inlineisa{shortest_path_pos_cost}
locale are implied by the checker precondition and witness predicate, and
we prove that our final theorem proved in that locale implies the checker
postcondition.

\subsection{Maximum Cardinality Matching in Graphs}
  \label{sec:case-study-matching}
  Our third case study is about maximum cardinality matching in general graphs.
A \emph{matching} in a graph $G$ is a subset $M$ of the edges of $G$ such
that no two share an endpoint.  A matching has maximum cardinality if its
cardinality is at least as large as that of any other matching.
Figure~\ref{fig:matching} shows a graph, a maximum cardinality matching,
and a witness of this fact.  An \emph{odd-set cover} $L$ of a
graph $G$ is a labeling of the vertices of $G$ with integers such that
every edge of $G$ is either incident to a vertex labeled 1 or connects two
vertices labeled with the same number $i$ and $i \ge 2$.

\begin{figure}[t]
\centering
\begin{tikzpicture}
\tikzstyle{circ}=[draw, fill=white, circle, inner sep=2pt, drop shadow]
\tikzstyle{norm}=[]
\tikzstyle{strong}=[line width=3pt]
\node[circ] (n0) at (0.7,2.1) {1};
\node[circ] (n1) at (0,1.4) {0};
\node[circ] (n2) at (1.4,1.4) {1};
\node[circ] (n3) at (2.8,1.4) {0};
\node[circ] (n4) at (4.2,1.4) {1};
\node[circ] (n5) at (3.5,0.7) {0};
\node[circ] (n6) at (0,0) {2};
\node[circ] (n7) at (1.4,0) {2};
\node[circ] (n8) at (2.8,0) {1};
\node[circ] (n9) at (4.2,0) {0};
\node[circ] (n10) at (0.7,-0.7) {2};
\node[circ] (n11) at (3.5,-0.7) {0};
\draw[strong] (n0) -- (n1);
\draw[norm] (n1) -- (n2);
\draw[norm] (n0) -- (n2);
\draw[strong] (n2) -- (n3);
\draw[norm] (n3) -- (n4);
\draw[norm] (n4) -- (n5);
\draw[norm] (n2) -- (n7);
\draw[norm] (n3) -- (n8);
\draw[strong] (n4) -- (n9);
\draw[strong] (n6) -- (n7);
\draw[norm] (n7) -- (n8);
\draw[norm] (n8) -- (n9);
\draw[norm] (n6) -- (n10);
\draw[norm] (n7) -- (n10);
\draw[strong] (n8) -- (n11);
\end{tikzpicture}
\caption{The vertex labels certify that the indicated matching is of
  maximum cardinality: All edges of the graph have either both endpoints
  labeled as 2 or at least one endpoint labeled as 1. Any matching can
  hence use at most one edge with both endpoints labeled 2 and at most four
  edges that have an endpoint labeled 1. Therefore, no matching has more
  than five edges. The matching shown consists of five edges (in bold).} 
  \label{fig:matching} 
\end{figure}
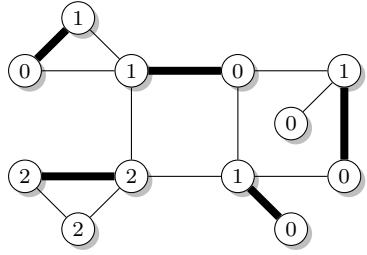

\begin{theorem}[Edmonds~\cite{edmonds:1965:matching}]
\label{thm:edm}
Let $M$ be a matching in a graph $G$, and let $L$ be an odd-set cover of
$G$.  For any $i \ge 0$, let $n_i$ be the number of vertices labeled $i$.
If
\begin{equation}
\abs{M} = n_1 + \sum_{i\ge 2}\lfloor n_i/2\rfloor,
\label{eq:cardM}
\end{equation}
then $M$ is a maximum cardinality matching.
\end{theorem}

\begin{proof}
Let $N$ be any matching in $G$.  For $i \ge 2$, let $N_i$ be the edges in
$N$ that connect two vertices labeled $i$, and let $N_1$ be the remaining
edges in $N$. Then, by the definition of odd-set cover, every edge in $N_1$
is incident to a vertex labeled 1. Since edges in a matching do not share
endpoints, we have
\[\abs{N_1} \le n_1\;\;\text{and}\;\;\abs{N_i} \le \lfloor n_i/2 \rfloor
\;\;\text{for $i \ge 2$.}\]
Thus, $\abs{N} \le  n_1 + \sum_{i\ge 2} \lfloor n_i / 2 \rfloor = \abs{M}$. 
\qed
\end{proof}

For every maximum cardinality
matching $M$ there is an odd-set cover $L$ satisfying
equality~(\ref{eq:cardM}); the proof of this is nontrivial and of no importance 
for the purpose of this paper. The cover uses nonnegative vertex labels in the
range $0$ to $\abs{V} - 1$ and all $n_i$'s with $i \ge 2$ are odd.  The
\emph{certifying algorithm for maximum cardinality matching} in LEDA
returns a matching $M$ and an odd-set cover $L$ such that (\ref{eq:cardM})
holds. The relationship to Section~\ref{sec:outline-of-approach} is as
follows:

\begin{center}
\begin{tabular}{r@{\hspace{1em}=\hspace{1em}}p{9cm}}
input $x$ & an undirected graph $G$ \\
output $y$ & a set of edges $M$ \\
witness $w$ & a vertex labeling $L$ \\
$\Pre(x)$ & $G$ and $M$ are wellformed and have no self-loops \\
$\Wit(x, y, w)$ & $M$ is a matching in $G$, $L$ is an odd-set cover for
$G$, and Equation~(\ref{eq:cardM}) holds \\
$\Post(x,y)$ & $M$ is a maximum cardinality matching in $G$.
\end{tabular}
\end{center}

Theorem~\ref{thm:edm} is the witness property. We give a formal proof for
it in Section~\ref{sec:matching-isabelle}. It is easy to write a correct program that
checks whether a set of edges is a matching and a vertex labeling is an
odd-set cover, which together satisfy Equation~\ref{eq:cardM}, see
Section~\ref{sec:matching-checker}.  In Section~\ref{sec:matching-vcc}, we
describe the verification of such a checker. The relevant files in the
companion website are check\_matching.c (C-code and checker correctness),
check\_matching.thy (Isabelle/HOL representation of the VCC checker
specification and witness property), Matching.thy (abstract verification of
the witness property in a locale), and Matching\_Link.thy (instantiation of
the locale with the Isabelle/HOL representation of the VCC specification).

This case study is a modified version of the one we present
in~\cite{alkassar:2011:vcc}.  The proof of the witness property is very
similar to the one published in~\cite{rizkallah:2011:matching} except that
it uses the Isabelle graph library, which has been developed meanwhile.

\subsubsection{Maximum-Cardinality-Matching Checker}
  \label{sec:matching-checker}
 
We build the checker using the graph data structure as in the previous case
studies (Listing~\ref{lst:c-graph}).  We assume that graphs are
wellformed and have neither self-loops nor duplicate edges.  We treat the
edges of a graph as undirected edges.  Matchings are also represented by
graphs.  We require an additional witness in the form of an array
\inlinevcc{f} that maps edge identifiers of the matching to edge identifiers
of the input graph.  For instance, if a graph consists of three edges
(identified as 0, 1 and 2) and the computed matching consists of the third
edge (\ie 2), then \inlinevcc{f} would be an array with a single element 2
indicating how the only edge of the matching corresponds to the edges of
the input graph.  Finally, the vertex labeling is represented by an array
\inlinevcc{osc}, which is indexed by vertices and stores elements of
type \inlinevcc{Nat}.
The checker function requires an auxiliary array \inlinevcc{check} that
can store as many elements of type \inlinevcc{Nat} as there are vertices in
the input graph, but at least two.  We expect that this array is allocated
elsewhere and given as input to the checker.

In addition to the checker function, there are four helper functions
(Listing~\ref{lst:matching-checker}).  The checker succeeds if the first
three of them succeed and if the fourth function returns a value that is
equal to the number of edges of the matching \inlinevcc{M}.  In short, the
helper functions perform the following tasks.  The function
\inlinevcc{check_subset} checks whether \inlinevcc{M} is a subgraph of
\inlinevcc{G} with respect to the mapping \inlinevcc{f}.  The function 
\inlinevcc{check_matching} checks that \inlinevcc{M} is indeed a
matching (contains no two edges that are incident).  The function
\inlinevcc{check_osc} checks whether the vertex labeling is an odd-set cover
and that vertex labels are in the range $\{0, \ldots,
\text{\inlinevcc{G->n}} - 1\}$.  Finally, the function \inlinevcc{weight}
computes the sum on the right-hand side of Equation~(\ref{eq:cardM}).  This
computation is optimized by first searching for the greatest vertex label,
which can be considerably smaller than the maximal $\text{\inlinevcc{G->n}}
- 1$, and then summing up partial sums only until this greatest label.  The
main checker function passes the auxiliary array \inlinevcc{check} to
\inlinevcc{check_matching} as the \inlinevcc{degree_in_M} argument and to
\inlinevcc{weight} as the \inlinevcc{count} argument.

\begin{listing}
\begin{vcclst}
bool check_subset(Graph* G, Graph* M, Nat* f)
{
    Edge_Id e;

    for (e = 0; e < M->m; e++)
    {
        if (f[e] >= G->m) return FALSE;
        if (M->es[e].s == G->es[f[e]].s && M->es[e].t == G->es[f[e]].t) continue;
        if (M->es[e].s == G->es[f[e]].t && M->es[e].t == G->es[f[e]].s) continue;
        return FALSE;
    }
    return TRUE;
}

bool check_matching(Graph* M, Nat* degree_in_M)
{
    Vertex v; Edge_Id e;

    for (v = 0; v < M->n; v++) degree_in_M[v] = 0;
    for (e = 0; e < M->m; e++)
    {
        if (degree_in_M[M->es[e].s] == 1 || degree_in_M[M->es[e].t] == 1) return FALSE;
        degree_in_M[M->es[e].s] = 1;
        degree_in_M[M->es[e].t] = 1;
    }
    return TRUE;
}

bool check_osc(Graph* G, Nat* osc)
{
    Edge_Id e; Vertex v, w;

    for (v = 0; v < G->n; v++) if (osc[v] >= G->n) return FALSE;
    for (e = 0; e < G->m; e++)
    {
        v = G->es[e].s;
        w = G->es[e].t;
        if (osc[v] == 1 || osc[w] == 1 || (osc[v] == osc[w] && osc[v] $\ge$ 2)) continue;
        return FALSE;
    }
    return TRUE;
}

Nat weight(Graph* G, Nat* osc, Nat* count)
{
    Vertex v; Nat c, s, max = 1, r = (G->n > 2) ? G->n : 2;

    for (c = 0; c < r; c++) count[c] = 0;
    for (v = 0; v < G->n; v++)
    {
        count[osc[v]] = count[osc[v]] + 1;
        if (osc[v] > max) max = osc[v];
    }
    s = count[1];
    for (c = 2; c < max + 1; c++) s += count[c] / 2;
    return s;
}
\end{vcclst}

\caption{Helper functions of the maximum-cardinality-matching checker.}
  \label{lst:matching-checker}
\end{listing}

\subsubsection{Checker Correctness}
  \label{sec:matching-vcc}

We build on the abstract graph data structure of
Listing~\ref{lst:spec-graph}.  We require that graphs are wellformed
and contain no self-loops:

\begin{vccquote}
forall \Edge_Id i; i < G.num_edges ==> G.edge[i].src !!= G.edge[i].trg
\end{vccquote}

\noindent
nor duplicate edges:

\begin{vccquote}
forall \Edge_Id i1, i2; i1 < G.num_edges &&& i2 < G.num_edges &&& i1 !!= i2 ==>
    G.edge[i1].src !!= G.edge[i2].src ||| G.edge[i1].trg !!= G.edge[i2].trg
\end{vccquote}

\noindent
An abstract vertex labeling \inlinevcc{L} is a mapping from vertices to
natural numbers.  The mapping \inlinevcc{f} from edge identifiers to edge
identifiers has a straightforward representation as an abstract mapping.  We
omit here, as in the previous case studies, the description of abstraction
functions from concrete to abstract values.

The witness predicate is a conjunction of four predicates, each related to
one of the helper functions in Listing~\ref{lst:matching-checker}.

\begin{description}
\item[\textsf{check\_subset}:] \inlinevcc{M} must be a subgraph of
  \inlinevcc{G} \wrt the edge mapping \inlinevcc{f}, \ie every edge of
  \inlinevcc{M} must also be an edge of \inlinevcc{G} modulo symmetry of
  edges:

\begin{vccquote}
forall \Edge_Id i; i < M.num_edges ==>
    f[i] < G.num_edges &&&
    (M.edge[i].src = G.edge[f[i]].src $\land$ M.edge[i].trg = G.edge[f[i]].trg |||
     M.edge[i].src = G.edge[f[i]].trg $\land$ M.edge[i].trg = G.edge[f[i]].src)
\end{vccquote}

\item[\textsf{check\_matching}:] \inlinevcc{M} must be a matching, \ie no
  two edges of \inlinevcc{M} have a vertex in common:

\begin{vccquote}
forall \Edge_Id i1, i2;
    i1 < M.num_edges &&& i2 < M.num_edges &&& i1 !!= i2 ==>
    M.edge[i1].src !!= M.edge[i2].src &&& M.edge[i1].src !!= M.edge[i2].trg &&&
    M.edge[i1].trg !!= M.edge[i2].src &&& M.edge[i1].trg !!= M.edge[i2].trg
\end{vccquote}

\item[\textsf{check\_osc}:] \inlinevcc{L} must be an odd-set cover of
  \inlinevcc{G}, \ie for every edge of \inlinevcc{G}, one of the edge's
  vertices is labeled 1 or both vertices are labeled by the same number
  greater than or equal to 2:

\begin{vccquote}
forall \Edge_Id i; i < G.num_edges ==>
    L[G.edge[i].src] = 1 |||
    L[G.edge[i].trg] = 1 |||
    L[G.edge[i].src] = L[G.edge[i].trg] &&& L[G.edge[i].src] >== 2
\end{vccquote}

\item[\textsf{weight}:] Equation~(\ref{eq:cardM}) must hold.  We define it
  stepwise.  The number of vertices labeled with \inlinevcc{c} is defined
  recursively:

\begin{vccquote}
_(def \natural \label_count(\Label L, \natural c, \natural i)
{
    return (i = 0) ? 0 : ((L[i - 1] = c) ? 1 : 0) + \label_count(L, c, i - 1);
})
\end{vccquote}

  \noindent
  We have $n_{\text{\inlinevcc{c}}} = \text{\inlinevcc{\\label_count(L, c,
  G.num_verts)}}$ for a vertex label \inlinevcc{c}.  The sum of these
  numbers for labels greater than 1 is again defined recursively:

\begin{vccquote}
_(def \natural \rec_weight(\Label L, \natural n, \natural i)
{
    return (i < 2) ? 0 : \label_count(L, i, n) / 2 + \rec_weight(L, n, i - 1);
})
\end{vccquote}

  \noindent
  We have $\sum_{i\ge 2} \lfloor n_i / 2 \rfloor =
  \text{\inlinevcc{\\rec_weight(L, G.num_verts, m)}}$ where \inlinevcc{m}
  is the greatest label assigned to any vertex by \inlinevcc{L}.
  The complete sum is then:

\begin{vccquote}
_(def \natural \full_weight(\Label L, \natural n, \natural i)
{
    return \label_count(L, 1, n) + \rec_weight(L, n, i);
})
\end{vccquote}

  \noindent
   That is, we have $n_1 + \sum_{i\ge 2} \lfloor n_i / 2 \rfloor =
   \text{\inlinevcc{\\full_weight(L, G.num_verts, m)}}$ with the same
   \inlinevcc{m} as before.  Finally, the predicate capturing
   Equation~(\ref{eq:cardM}) is as follows:

\begin{vccquote}
M.num_edges = \full_weight(L, G.num_verts, m) &&&
forall \Vertex v; v < G.num_verts ==> L[v] <== m
\end{vccquote}

\end{description}

Verifying the correctness of the checker
(Section~\ref{sec:matching-checker}) is done in the same way as the earlier case
studies for the first three predicates above.  We only have to provide the
right loop invariants, and simple variations of the predicates to be proved
are sufficient.  In \inlinevcc{check_matching}, we need additional loop
invariants.  Along with the first loop, we accumulate the knowledge about
the initialization of the \inlinevcc{degree_in_M} array by specifying that
the first positions of the array have already been set to \inlinevcc{0}:

\begin{vccquote}
forall Nat u; u < v ==> degree_in_M[u] = 0
\end{vccquote}

\noindent
Moreover, on the second loop, we need three additional loop invariants. One
invariant states that values stored in \inlinevcc{degree_in_M} are in
range:

\begin{vccquote}
forall Nat v; v < M->n ==> degree_in_M[v] $\le$ 1
\end{vccquote}

\noindent
Another invariant states that vertices, for which
\inlinevcc{degree_in_M} is still 0, cannot be part
of any already checked edge:

\begin{vccquote}
forall Nat v; v < M->n &&& degree_in_M[v] = 0 ==>
    forall Nat e1; e1 < e ==> M->es[e1].s !!= v &&& M->es[e1].t !!= v
\end{vccquote}

\noindent
Finally, vertices for which \inlinevcc{degree_in_M} has already been
set to 1 are mapped by a ghost mapping \inlinevcc{E} to their adjacent edge
in the matching \inlinevcc{M}:

\begin{vccquote}
forall Vertex v; v < M->n &&& degree_in_M[v] = 1 ==>
    E[v] < e &&& (M->es[E[v]].s = v ||| M->es[E[v]].t = v)
\end{vccquote}

\noindent
This invariant is required to prove completeness. We maintain this
invariant by updating the ghost mapping \inlinevcc{E} in the loop body
accordingly.

%

Proving the \inlinevcc{weight} function correct is the most intricate part
of the checker verification.  There are two properties that need to be shown:\ functional
correctness and the absence of overflows.  The former requires that the
function computes the $n_i$ and the overall sum of
Equation~(\ref{eq:cardM}) correctly, as specified in the above fourth
conjunct of the witness predicate.  The latter requires that the
additions in both the second and third loop do not overflow.  Surprisingly,
the absence of overflows is much harder to establish than functional
correctness.

We concentrate first on functional correctness.
The second loop updates the \inlinevcc{count} array in a way that
maintains the following property:

\begin{vccquote}
forall Nat j; j < r ==> count[j] = \label_count(L, j, v)
\end{vccquote}

\noindent
From this property follows this loop invariant on the third loop:

\begin{vccquote}
s = \full_weight(L, G->n, c - 1))
\end{vccquote}

\noindent
Together with a further loop invariant for the second loop to guarantee
that \inlinevcc{max} is the greatest label seen so far, we can conclude
that the \inlinevcc{weight} function is functionally correct.

The addition in the second loop can never overflow because in each loop
iteration, the loop variable is an upper limit on the value 
\inlinevcc{count[i]} for each label \inlinevcc{i}.  Concerning the
addition in the third loop, we observe that in each loop iteration, the
value of \inlinevcc{s} is bounded by the number of vertices in
\inlinevcc{G}.  To establish this property, we build up a ghost map
\inlinevcc{sum} in the second loop in such a way that in every iteration
of that loop, this map fulfills the following invariant:

\begin{vccquote}
sum[1] = count[1] $\land$
(forall Nat j; 1 < j &&& j < r ==> sum[j] = sum[j - 1] + count[j]) &&&
(forall Nat j; 1 < j &&& j < r ==> sum[j] <== v)
\end{vccquote}

\noindent
Maintaining this invariant requires updating the \inlinevcc{sum} map during
each iteration of the second loop.  We do so in a nested ghost loop in which
we propagate the increment that happened on the \inlinevcc{count} array to
every possibly affected element \inlinevcc{sum[j]}.

The postcondition of the checker expresses that the cardinality of any
matching of \inlinevcc{G} cannot be smaller than the cardinality of
\inlinevcc{M}:

\begin{vccquote}
forall \Graph M2; \Edge_Map I2; \is_subset(G, M2, I2) &&& \is_matching(M2) ==>
    M2.num_edges <== M.num_edges
\end{vccquote}

We give a formal proof that the checker preconditions and the witness
predicate imply this property in the following section.

\subsubsection{Proof of the Witness Property for the
  Maximum-Cardinality-Matching Checker}
  \label{sec:matching-isabelle}

We explain the Isabelle proof for the witness property, \ie for 
Theorem~\ref{thm:edm}. See Listing~\ref{lst:matching-locale} for an excerpt
of our formal Isabelle proof development that can be found in file
Matching.thy. The formal proof follows the scheme of the textbook proof
and is split into two main parts.

\begin{listing}
\begin{isalst}
type_synonym label = nat

definition disjoint_edges :: "('a, 'b) pre_graph \<Rightarrow> 'b \<Rightarrow> 'b \<Rightarrow> bool" where
    "disjoint_edges \G e1 e2 = (
         start \G e1 \<noteq> start \G e2 \<and> start \G e1 \<noteq> target \G e2 \<and> 
         target \G e1 \<noteq> start \G e2 \<and> target \G e1 \<noteq> target \G e2)"

definition matching :: "('a, 'b) pre_graph \<Rightarrow> 'b set \<Rightarrow> bool" where
    "matching \G \M = (
         \M \<subseteq> edges \G \<and>
         (\<forall>e1 \<in> \M. \<forall>e2 \<in> \M. e1 \<noteq> e2 \<longrightarrow> disjoint_edges \G e1 e2))"

definition OSC :: "('a, 'b) pre_graph \<Rightarrow> ('a \<Rightarrow> label) \<Rightarrow> bool" where
    "OSC \G \L = (
         \<forall>\e \<in> edges \G.
             \L (start \G \e) = \1 \<or> \L (target \G \e) = \1 \<or> 
             \L (start \G \e) = \L (target \G \e) \<and> \L (start \G \e) \<ge> \2)"

definition weight :: "label set \<Rightarrow> (label \<Rightarrow> nat) \<Rightarrow> nat" where
    "weight LV \f = \f \1 + \<Sum>\i \<in> LV. (\f \i) div \2"

definition \N :: "'a set \<Rightarrow> ('a \<Rightarrow> label) \<Rightarrow> label \<Rightarrow> nat" where
    "\N \V \L \i = card {\v \<in> \V. \L \v = \i}"

locale matching_locale = digraph +
    fixes maxM :: "'b set"
    fixes \L :: "'a \<Rightarrow> label"
    assumes \matching: "matching G maxM" 
    assumes \OSC: "OSC \G \L"
    assumes \weight: "card maxM = weight {\i \<in> \L ` verts \G. \i > \1} (\N (verts \G) \L)"
\end{isalst}
\caption{Definitions and locale for the matching proof in Isabelle.}
\label{lst:matching-locale}
\end{listing}

For $i \ge 2$, let $M_i$ be the edges in $M$ that connect two vertices
labeled $i$, and let $M_1$ be the remaining edges in $M$. The sets $M_i$, $i \ge 1$, are disjoint. 
We use the
definition of an odd-set cover to prove $M \subseteq  \bigcup_{i\ge1}
M_i$, and thus $\abs{M} \le \sum_{i\ge1}\abs{M_i}$ by disjointness of the sets $M_i$.  Let $V_i$ be the
vertices labeled~$i$, and let $n_i = \abs{V_i}$.  We formally prove:
$\abs{M_1} \le n_1$ and $\abs{M_i} \le \lfloor n_i/2 \rfloor$.

\newcommand{\pinV}{\mathit{endpoint}_{V_1}}

In order to prove $\abs{M_1} \le n_1$, we exhibit an injective function
from $M_1$ to $V_1$. We first prove, using the definition of an odd-set
cover, that every edge $e \in M_1$ has at least one endpoint in $V_1$. This
gives rise to a function $\pinV$ that maps from $M_1$ to $V_1$. We then use 
that edges in a matching do not share endpoints (i.e., edges in a
matching are disjoint when interpreted as sets) to conclude that $\pinV$
is injective. This establishes $\abs{M_1} \le \abs{V_i}$.

For $i\ge2$, the proof of the inequality $\abs{M_i} \le \lfloor n_i/2
\rfloor$ is similar but more involved. $M_i$ is a set of edges.  If we
represent edges as sets, each with cardinality two, then $M_i$ is a
collection of sets.  We define the set of vertices $V^\prime_i$ to be
$\bigcup_{i\ge2} M_i$ and use the definition of an odd-set cover to prove
$V^\prime_i \subseteq V_i$.  Since the edges in a
matching are pairwise disjoint, we obtain $\abs{V^\prime_i} = 2 \cdot
\abs{M_i}$.  Note also that $\abs{V^\prime_i}$ must be even since
$\abs{M_i}$ is a natural number.  Thus, we can prove that $\abs{M_i} \le
\left\lfloor\abs{V^\prime_i}/\,2\right\rfloor$, and hence, $\abs{M_i} \le
\left\lfloor\abs{V^\prime_i} /\,2\right\rfloor \le \left\lfloor
\abs{V_i}/\,2 \right\rfloor = \lfloor n_i/2 \rfloor$.

Instantiating this Isabelle proof for the data structures and properties
exported from VCC is fairly straightforward since both formalizations have
been chosen intentionally close to each other. We prove by induction that
\mbox{$N\ \{0\ ..\!\!< n\}\ L\ l$} equals \inlinevcc{\\label_count(L, l, n)} for
every label $l$. Moreover, we prove by induction that $\mathit{weight}\ \{2\ ..\ k\}\ f$ equals
\inlinevcc{\\full_weight(L, n, k)} if $f\ l$ and
\inlinevcc{\\label_count(L, l, n)} coincide. After showing that $\abs{M}$
equals \inlinevcc{M.num_edges}, we can establish the witness property for
the matching checker.

\section{Evaluation}
  \label{sec:evaluation}

We have shown that our methodology allows us to lift the trustworthiness of
certifying algorithms to a new level: one can obtain formal instance
correctness of certifying algorithms. A certifying algorithm returns for
any  input $\overline{x}$ satisfying the precondition $\Pre$ a pair
$(\overline{y},\overline{w})$. It is accompanied by a witness 
predicate $\Wit$ and an associated checker program $C$. If $\overline{x}$
satisfies the precondition, $C$ is supposed to halt on input
$(\overline{x}, \overline{y},\overline{w})$ and decide on the witness
predicate. We have shown in three cases that checker correctness can be
established formally using VCC. Precondition plus witness predicate are
supposed to imply the postcondition $\Post$. We have shown in the same
three cases that the witness property can be established formally using
Isabelle/HOL. In all cases, we followed the same approach: (1) write the
checker in C and annotate it so that VCC can establish that the checker
decides the witness property, (2) translate precondition, witness predicate,
and postcondition from VCC to Isabelle/HOL, and (3) prove the witness
property in Isabelle/HOL and add the VCC-version of it as an axiom to VCC.
We translated manually from VCC to Isabelle/HOL.  Since the translation is
purely syntactical, it can be automated.

We believe that our approach would work for all certifying algorithms
discussed in the survey on certifying
algorithms~\cite{mcconnell-etal:2011:certifying}.  This might require 
formalizing the subject matter in Isabelle/HOL first. Our work profited from
the recent formalization of graphs in 
Isabelle/HOL~\cite{noschinski:2011:graph-library}. A similar effort would
be necessary before geometric or randomized algorithms can be tackled.

It is always a gratifying experience to see that a formalization reveals
gaps in paper and pen reasoning. In our case, we found that the
checker for a maximum cardinality matching in LEDA does not verify that the
computed matching $M$ is a subgraph of $G$.

The matching algorithm for general graphs and its efficient implementation
are advanced topics in graph algorithms. The algorithm is highly complex 
and is not covered in the standard textbooks on algorithms.
The following page numbers illustrate the complexity gap between the
original algorithm and the checker: In the LEDA book, the description of
the algorithm for computing the maximum cardinality matching and the proof
of its correctness fills about 15 pages, compared to a one-page description
of the checker implementation.

All described theorems and lemmas have been formally verified using 
VCC and Isabelle/HOL. Table~\ref{tab:effort} summarizes our effort in 
lines of code, lines of annotations, and lines of Isabelle specifications and
proofs.  We observe an average ratio of 2.5 for the VCC annotation overhead, 
which includes our specifications for the witness predicates besides annotations 
for code verification. The overhead grows to an average ratio of 6.7 when 
also considering our Isabelle proofs. We believe that this overhead is a small 
price to pay for verifying full functional
correctness of code with nontrivial mathematical properties. Observe that
the Isabelle linking proofs are much shorter than the high-level proofs in
the Isabelle locales for the two more complex checkers. Linking proofs,
hence, cause only a small overhead and, as we found, are fairly 
straightforward.  In our opinion, splitting the proofs into high-level
proofs for mathematical concepts and lower-level linking proofs
improved productivity and certainly helped to reduce the overall proof size
and effort.  The overall VCC proof time for all checkers is 10 seconds on a
2.9 GHz Intel Core i7 machine. The overall Isabelle proof time is also well
below one minute on the same machine.

\begin{table}
\centering
\begin{tabular}{l@{\hspace{2.5em}}c@{\hspace{2.5em}}c@{\hspace{2.5em}}ccc}
& & \multicolumn{1}{c@{\hspace{3em}}}{VCC} & \multicolumn{3}{c}{Isabelle} \\[0.5ex]
& C code & Annotations & Exports & Locales & Links \\[3pt]
\hline\\[-0.5em]
Connected components & \phantom{0}61 & 162 & \phantom{0}63 & 102 & 163 \\[0.5ex]
Shortest paths       & 107 & 318 & 109 & 279 & 198 \\[0.5ex]
Matching             & 124 & 263 & \phantom{0}78 & 318 & 164 \\[0.5em]
\hline
\end{tabular}


\caption{Lines of C code, annotations for VCC, and Isabelle declarations
  and proofs, excluding empty lines in all cases. For the latter, we give the 
  numbers for the specifications exported from VCC, abstract proofs performed 
  within locales, and proofs that link the concrete representation exported from
  VCC with the abstract proofs separately.}
  \label{tab:effort}
\end{table}

It took several months to develop the framework and complete the first example
as described in~\cite{alkassar:2011:vcc}. For this paper, we have reworked the
framework, thereby strengthening and simplifying it at the same time. We
now prove total correctness of the checkers and not only partial
correctness. Moreover, we now establish that the checker accepts a triple
$(x,y,w)$ if and only if $w$ is a valid witness for output $y$, whereas
previously, we only proved one direction. The simplification results from
dropping concrete specifications in VCC.  
Moreover, a new version of VCC has allowed for shorter specifications and proofs.

Our methodology is not confined to verifying certifying computations only.
Rather, any proposition that VCC fails to establish can be exported to
Isabelle and proved there; see \cite[Section 4.5]{boehme:2012:thesis} for an 
example. Hence, our methodology extends the applicability of VCC and allows, in
conjunction with Isabelle, the verification of complicated problems that were 
infeasible up to now.

\section{Related Work}
  \label{sec:related-work}

The notion of a certifying algorithm is ancient. Al-Khawarizmi 
already described how to (partially) check the correctness of a multiplication in
his book on algebra. The extended Euclidean algorithm for greatest common
divisors is also certifying; it dates back to the 17th century.  Yet, formal
verification of checkers is recent.

In 1997, Bright et al.~\cite{bright-etal:1997:certifier} verified a checker
for a sorting algorithm that has been formalized in the Boyer-Moore theorem
prover~\cite{boyer-moore:1990:prover}. De Nivelle and Piskac formally
verified the checker for priority queues implemented in
LEDA~\cite{nivelle-piskac:2005:checker}.  Bulwahn et
al.~\cite{bulwahn-etal:2008:imperative} describe a verified SAT checker,
\ie a checker for certificates of unsatisfiability produced by a SAT
solver.  They develop the checker and prove its correctness within Isabelle/HOL.
Similar proof checkers have been formalized in the
Coq~\cite{bertot-etal:2004:coq} proof
assistant~\cite{darbari-etal:2010:sat,armand-etal:2010:sat}.
CeTA~\cite{thiemann-sternagel:2009:ceta}, a tool for certified termination
analysis, is also based on formally verified checkers.  In contrast to our approach, all mentioned checkers are
entirely developed and verified within the language of a theorem prover.

VCC has been applied to verify tens of thousands of lines of C code. So far, the majority of its verification targets have been restricted to system-level code from the
domain of microkernels and
hypervisors~\cite{Baumann-etal:2009:PikeOS,klein-etal:2010:sel4,Shi-etal:2012:Orientais}.
Our work extends the range of VCC applications to graph algorithms and, in
general, to any code that requires nontrivial mathematical reasoning to
establish full functional correctness.

In contrast to using an automatic program verifier such as VCC as we did,
imperative code can also be verified in interactive theorem provers such as
Coq, HOL~\cite{gordon-melham:1993:hol}, or Isabelle/HOL.  This requires a
formalization of the imperative language and its semantics within the
theorem prover.
Chargu{\'e}raud describes a tool, CFML, that is based on the idea of
extracting characteristic formulae from
programs~\cite{chargueraud:2011:cfml}. CFML is embedded in Coq and targets
imperative Caml programs. It has been applied to verify imperative data
structures such as mutable lists, sparse arrays and
union-find.
Norrish presented a formal semantics of C formalized in the HOL theorem
prover~\cite{norrish:1998:c}. Parallel to this work, a subset of C, called
C0, was formalized in Isabelle/HOL~\cite{leinenbach-etal:2005:c0}. Schirmer 
developed a verification environment for sequential imperative programs 
within Isabelle and embedded C0 into this
environment~\cite{schirmer:2006:thesis};  his verification environment is written in 
the generic imperative programming language Simpl. Shirmer's work has been applied, for
instance, to verify a compiler for C0~\cite{petrova:2007:thesis}.
More recently, the seL4 microkernel, written in C, has been
verified in Isabelle/HOL~\cite{klein-etal:2010:sel4} based on the work of
Norrish and Schirmer. The underlying approach is refinement; one 
verifies an intermediate implementation in Haskell against the abstract 
specification, then generates C code from the Haskell code using a verified 
generator. The reverse process of turning C code into a 
 human readable abstraction can partly be
automated~\cite{greenaway-etal:2012:gap}. The latter work opens up an
alternative path for the verification of certifying computations. One
writes the checker in C, uses the trusted translation to obtain an equivalent
Simpl program, and then proves checker correctness and the witness property in
Isabelle/HOL. We plan to also try this alternative route. There are advantages and
disadvantages for both approaches, and it is too early to claim the superiority of one
approach over the other. An advantage of our current approach is that we annotate 
the source program; a disadvantage is that we use two tools. An advantage of the 
alternative approach is that only one tool is used; the disadvantage then is that one 
cannot argue about the source program but rather about a program obtained by 
a translation process.

Previous work that proposes, as we do, the use of interactive theorem provers as
backends to code verification systems comprises, for instance, the link
between Boogie and Isabelle/HOL~\cite{boehme-etal:2010:holboogie} and the
link between Why and Coq~\cite{filliatre-marche:2007:why}.
Both systems have a C verifier frontend.  Such approaches for
connecting code verifiers and proof assistants usually give the latter the same
information that is made available to the first-order engine, overwhelming
the users of the proof assistants with a mass of detail.  Instead, we allow
only clean chunks of mathematics to move between the verifier and the proof
assistant. This hides details of the underlying
programming languages from the proof assistant, thus requiring the user to discharge
only interesting proof obligations.

Shortest-path algorithms, especially imperative implementations thereof,
are popular as case studies for demonstrating code
verification~\cite{chargueraud:2011:cfml,nordhoff-lammich:2010:dijkstra,boehme-etal:2008:holboogie}.
They target full functional correctness as opposed to
instance correctness. Verifying instance correctness is orthogonal to verifying
the implementation of a particular shortest path algorithm. 
Our work is directly applicable to any implementation of shortest-path that is instrumented to provide the necessary witness expected by our checker.

To our knowledge, there has been no other attempt to verify
algorithms or checkers for connected components or maximum cardinality
matchings.

\section{Conclusion and Future Work}
  \label{sec:conclusions}

We have described a framework for the verification of certifying computations and
applied it to three nontrivial combinatorial problems: connectivity of
graphs, shortest paths in graphs, and maximum cardinality matchings in
graphs. Our work greatly increases the trustworthiness of certifying algorithms.

Specifically, for each instance of the considered three problems, we can
now give a formal proof of the correctness of the result. Thus, the user
has neither to trust the implementation of the original algorithm nor the
checker, nor does he have to understand why the witness property holds. 
We
stress that we did not prove the correctness of the original programs but
only verified the results of their computations.

Our methodology
can be applied to any other problem for which a certifying algorithm is
known; see~\cite{mcconnell-etal:2011:certifying} for a survey.
The fourth author is currently verifying a checker for shortest paths with arbitrary edge costs~\cite{rizkallah:2012:shortest}, and 
Lars Noschinski is verifying a checker for graph non-planarity testing.

Our methodology is not restricted to verifying certifying computations.
The integration of VCC and Isabelle/HOL should be useful whenever
verification of a program requires nontrivial mathematical reasoning.

\begin{acknowledgements}
We thank Ernie Cohen for his advice on VCC idioms, Mark Hillebrand for
insightful discussions on VCC, and Norbert Schirmer for his initial Isabelle
support.  Moreover, we thank Lars Noschinski for developing a powerful graph
library in Isabelle/HOL. We are grateful to Jasmin Christian Blanchette for providing 
useful comments on an earlier version of this article.
\end{acknowledgements}

\bibliographystyle{spmpsci}
\bibliography{journal}

\begin{thebibliography}{10}
\providecommand{\url}[1]{{#1}}
\providecommand{\urlprefix}{URL }
\expandafter\ifx\csname urlstyle\endcsname\relax
  \providecommand{\doi}[1]{DOI~\discretionary{}{}{}#1}\else
  \providecommand{\doi}{DOI~\discretionary{}{}{}\begingroup
  \urlstyle{rm}\Url}\fi

\bibitem{alkassar:2011:vcc}
Alkassar, E., B{\"o}hme, S., Mehlhorn, K., Rizkallah, C.: Verification of
  certifying computations.
\newblock In: Computer Aided Verification, \emph{Lecture Notes in Computer
  Science}, vol. 6806, pp. 67--82. Springer (2011)

\bibitem{armand-etal:2010:sat}
Armand, M., Gr{\'e}goire, B., Spiwack, A., Th{\'e}ry, L.: Extending {Coq} with
  imperative features and its application to {SAT} verification.
\newblock In: Interactive Theorem Proving, \emph{Lecture Notes in Computer
  Science}, vol. 6172, pp. 83--98. Springer (2010)

\bibitem{barnett-etal:2006:boogie}
Barnett, M., Chang, B.Y.E., DeLine, R., Jacobs, B., Leino, K.R.M.: Boogie: A
  modular reusable verifier for object-oriented programs.
\newblock In: Formal Methods for Components and Objects, \emph{Lecture Notes in
  Computer Science}, vol. 4111, pp. 364--387. Springer (2006)

\bibitem{Baumann-etal:2009:PikeOS}
Baumann, C., Beckert, B., Blasum, H., Bormer, T.: Formal verification of a
  microkernel used in dependable software systems.
\newblock In: Computer Safety, Reliability, and Security, \emph{Lecture Notes
  in Computer Science}, vol. 5775, pp. 187--200. Springer (2009)

\bibitem{bertot-etal:2004:coq}
Bertot, Y., Cast{\'e}ran, P.: Interactive Theorem Proving and Program
  Development---{Coq'Art}: The Calculus of Inductive Constructions.
\newblock Texts in Theoretical Computer Science. An EATCS Series. Springer
  (2004)

\bibitem{blum-kannan:1989:check}
Blum, M., Kannan, S.: Designing programs that check their work.
\newblock In: Symposium on Theory of Computing, pp. 86--97. ACM (1989)

\bibitem{boehme:2012:thesis}
B\"ohme, S.: Proving theorems of higher-order logic with {SMT} solvers.
\newblock Ph.D. thesis, Technische Universit\"at M\"unchen (2012)

\bibitem{boehme-etal:2008:holboogie}
B{\"o}hme, S., Leino, K.R.M., Wolff, B.: {HOL-Boogie}---{A}n interactive prover
  for the {Boogie} program-verifier.
\newblock In: Theorem Proving in Higher Order Logics, \emph{Lecture Notes in
  Computer Science}, vol. 5170, pp. 150--166. Springer (2008)

\bibitem{boehme-etal:2010:holboogie}
B{\"o}hme, S., Moskal, M., Schulte, W., Wolff, B.: {HOL-Boogie}---an
  interactive prover-backend for the {V}erifying {C} {C}ompiler.
\newblock Journal of Automated Reasoning \textbf{44}(1--2), 111--144 (2010)

\bibitem{boyer-moore:1990:prover}
Boyer, R.S., Moore, J.S.: A theorem prover for a computational logic.
\newblock In: Conference on Automated Deduction, \emph{Lecture Notes in
  Computer Science}, vol. 449, pp. 1--15. Springer (1990)

\bibitem{bright-etal:1997:certifier}
Bright, J.D., Sullivan, G.F., Masson, G.M.: A formally verified sorting
  certifier.
\newblock IEEE Transactions on Computers \textbf{46}(12), 1304--1312 (1997)

\bibitem{bulwahn-etal:2008:imperative}
Bulwahn, L., Krauss, A., Haftmann, F., Erk{\"o}k, L., Matthews, J.: Imperative
  functional programming with {I}sabelle/{HOL}.
\newblock In: Theorem Proving in Higher Order Logics, \emph{Lecture Notes in
  Computer Science}, vol. 5170, pp. 134--149. Springer (2008)

\bibitem{chargueraud:2011:cfml}
Chargu{\'e}raud, A.: Characteristic formulae for the verification of imperative
  programs.
\newblock In: International Conference on Functional Programming, pp. 418--430.
  ACM (2011)

\bibitem{cohen-etal:2009:vcc}
Cohen, E., Dahlweid, M., Hillebrand, M., Leinenbach, D., Moskal, M., Santen,
  T., Schulte, W., Tobies, S.: {VCC}: {A} practical system for verifying
  concurrent {C}.
\newblock In: Theorem Proving in Higher Order Logics, \emph{Lecture Notes in
  Computer Science}, vol. 5674, pp. 23--42. Springer (2009)

\bibitem{darbari-etal:2010:sat}
Darbari, A., Fischer, B., Marques-Silva, J.: Industrial-strength certified
  {SAT} solving through verified {SAT} proof checking.
\newblock In: Theoretical Aspects of Computing, \emph{Lecture Notes in Computer
  Science}, vol. 6255, pp. 260--274. Springer (2010)

\bibitem{edmonds:1965:matching}
Edmonds, J.: Maximum matching and a polyhedron with 0,1-vertices.
\newblock Journal of Research of the National Bureau of Standards \textbf{69B},
  125--130 (1965)

\bibitem{filliatre-marche:2007:why}
Filli{\^a}tre, J.C., March{\'e}, C.: The {Why/Krakatoa/Caduceus} platform for
  deductive program verification.
\newblock In: Computer Aided Verification, \emph{Lecture Notes in Computer
  Science}, vol. 4590, pp. 173--177. Springer (2007)

\bibitem{gordon-etal:1979:lcf}
Gordon, M., Milner, R., Wadsworth, C.P.: {Edinburgh LCF}: A Mechanised Logic of
  Computation, \emph{Lecture Notes in Computer Science}, vol.~78.
\newblock Springer (1979)

\bibitem{gordon-melham:1993:hol}
Gordon, M.J.C., Melham, T.F. (eds.): Introduction to HOL: A Theorem-Proving
  Environment for Higher-Order Logic.
\newblock Cambridge University Press (1993)

\bibitem{greenaway-etal:2012:gap}
Greenaway, D., Andronick, J., Klein, G.: Bridging the gap: Automatic verified
  abstraction of {C}.
\newblock In: Interactive Theorem Proving, \emph{Lecture Notes in Computer
  Science}, vol. 7406, pp. 99--115. Springer (2012)

\bibitem{klein-etal:2010:sel4}
Klein, G., Andronick, J., Elphinstone, K., Heiser, G., Cock, D., Derrin, P.,
  Elkaduwe, D., Engelhardt, K., Kolanski, R., Norrish, M., Sewell, T., Tuch,
  H., Winwood, S.: {seL4}: Formal verification of an operating-system kernel.
\newblock Communications of the ACM \textbf{53}(6), 107--115 (2010)

\bibitem{leinenbach-etal:2005:c0}
Leinenbach, D., Paul, W.J., Petrova, E.: Towards the formal verification of a
  {C0} compiler: Code generation and implementation correctness.
\newblock In: Software Engineering and Formal Methods, pp. 2--12. IEEE Computer
  Society (2005)

\bibitem{mcconnell-etal:2011:certifying}
McConnell, R.M., Mehlhorn, K., N{\"a}her, S., Schweitzer, P.: Certifying
  algorithms.
\newblock Computer Science Review \textbf{5}(2), 119--161 (2011)

\bibitem{melhorn-naeher:1999:leda}
Mehlhorn, K., N\"aher, S.: The {LEDA} Platform for Combinatorial and Geometric
  Computing.
\newblock Cambridge University Press (1999)

\bibitem{moura-bjorner:2008:z3}
de~Moura, L.M., Bj{\o}rner, N.: {Z3}: An efficient {SMT} solver.
\newblock In: Tools and Algorithms for the Construction and Analysis of
  Systems, \emph{Lecture Notes in Computer Science}, vol. 4963, pp. 337--340.
  Springer (2008)

\bibitem{nipkow-etal:2002:isabelle}
Nipkow, T., Paulson, L.C., Wenzel, M.: Isabelle/HOL --- A Proof Assistant for
  Higher-Order Logic, \emph{Lecture Notes in Computer Science}, vol. 2283.
\newblock Springer (2002)

\bibitem{nivelle-piskac:2005:checker}
de~Nivelle, H., Piskac, R.: Verification of an off-line checker for priority
  queues.
\newblock In: Software Engineering and Formal Methods, pp. 210--219. IEEE
  Computer Society (2005)

\bibitem{nordhoff-lammich:2010:dijkstra}
Nordhoff, B., Lammich, P.: Dijkstra's shortest path algorithm.
\newblock Archive of Formal Proofs \textbf{2012} (2012).
\newblock \url{http://afp.sourceforge.net/entries/Dijkstra_Shortest_Path.shtml}

\bibitem{norrish:1998:c}
Norrish, M.: C formalised in {HOL}.
\newblock Ph.D. thesis, Computer Laboratory, University of Cambridge (1998)

\bibitem{noschinski:2011:graph-library}
Noschinski, L.: A formalization of graphs in {Isabelle/HOL} (2011).
\newblock Ongoing work, available from author upon request

\bibitem{petrova:2007:thesis}
Petrova, E.: Verification of the {C0} compiler implementation on the source
  code level.
\newblock Ph.D. thesis, Saarland University, Saarbr{\"u}cken (2007)

\bibitem{rizkallah:2011:matching}
Rizkallah, C.: Maximum cardinality matching.
\newblock Archive of Formal Proofs \textbf{2011} (2011).
\newblock \url{http://afp.sourceforge.net/entries/Max-Card-Matching.shtml}

\bibitem{rizkallah:2012:shortest}
Rizkallah, C.: An axiomatic characterization of shortest-path in {Isabelle/HOL}
  (2012).
\newblock Formal proof development,
  \url{http://www.mpi-inf.mpg.de/~crizkall/Publications/shortest_path.pdf}

\bibitem{schirmer:2006:thesis}
Schirmer, N.: Verification of sequential imperative programs in {Isabelle/HOL}.
\newblock Ph.D. thesis, Technische Universit{\"a}t M{\"u}nchen (2006)

\bibitem{Shi-etal:2012:Orientais}
Shi, J., He, J., Zhu, H., Fang, H., Huang, Y., Zhang, X.: {ORIENTAIS}: Formal
  verified {OSEK/VDX} real-time operating system.
\newblock In: Engineering of Complex Computer Systems, pp. 293--301. IEEE
  Computer Society (2012)

\bibitem{sullivan-masson:1990:certification}
Sullivan, G.F., Masson, G.M.: Using certification trails to achieve software
  fault tolerance.
\newblock In: Fault-Tolerant Computing, pp. 423--431. IEEE Computer Society
  (1990)

\bibitem{thiemann-sternagel:2009:ceta}
Thiemann, R., Sternagel, C.: Certification of termination proofs using {CeTA}.
\newblock In: Theorem Proving in Higher Order Logics, \emph{Lecture Notes in
  Computer Science}, vol. 5674, pp. 452--468. Springer (2009)

\bibitem{verisoftxt}
Verisoft {XT}.
\newblock \url{http://www.verisoftxt.de} (2010)

\end{thebibliography}

\end{document}